\let\ifsub\iftrue \let\ifappendix\iftrue
\ifcsname ifdraft\endcsname\else
  \expandafter\let\csname ifdraft\expandafter\endcsname
                  \csname iffalse\endcsname
\fi

\ifcsname ifappendix\endcsname\else
  \expandafter\let\csname ifappendix\expandafter\endcsname
                  \csname iffalse\endcsname
\fi

\ifcsname ifsub\endcsname\else
  \expandafter\let\csname ifsub\expandafter\endcsname
                  \csname iffalse\endcsname
\fi

\ifcsname ifarxiv\endcsname\else
  \expandafter\let\csname ifarxiv\expandafter\endcsname
                  \csname iffalse\endcsname
\fi

\ifcsname ifwatermark\endcsname\else
  \expandafter\let\csname ifwatermark\expandafter\endcsname
                  \csname iffalse\endcsname
\fi

\ifsub
\documentclass[acmsmall,screen,authorversion]{acmart}
\else
\documentclass[acmsmall]{acmart}
\fi
\setcopyright{rightsretained}
\acmPrice{}
\acmDOI{10.1145/3434286}
\acmYear{2021}
\copyrightyear{2021}
\acmSubmissionID{popl21main-p25-p}
\acmJournal{PACMPL}
\acmVolume{5}
\acmNumber{POPL}
\acmArticle{5}
\acmMonth{1}

\usepackage{booktabs}   
\usepackage{subcaption} 

\usepackage{xspace}
\usepackage{dsfont}
\usepackage{tikz}
\usetikzlibrary{cd}
\usetikzlibrary{decorations.text}
\usetikzlibrary{shapes.multipart}
\usepackage{bussproofs}

\usepackage{microtype}
\usepackage{cancel}
\ifwatermark
\usepackage{draftwatermark}
\SetWatermarkText{Draft -- please do not distribute}
\SetWatermarkScale{1.5}
\SetWatermarkColor[gray]{0.8}
\SetWatermarkFontSize{1cm}
\fi

\newcommand{\specK}[1]{\ensuremath{\textcolor{blue}{#1}}}
\newcommand{\spec}[1]{\specK{\{{#1}\}}}

\newcommand{\proofspec}[1]{\ensuremath{{\color{blue}\{{#1}\}}}}
\newcommand{\frameC}[1]{\ensuremath{{\color{magenta!80!blue}{#1}}}}
\newcommand{\proofspecL}[1]{\ensuremath{{\color{blue}\{{#1}}}}

\newcommand{\res}[1]{\ensuremath{\mathsf{#1}}\xspace}
\newcommand{\prog}[1]{\ensuremath{\mathsf{#1}}\xspace}
\newcommand{\statesp}[1]{\ensuremath{\Sigma_{#1}}}

\newcommand{\trans}[1]{\ensuremath{\mathsf{#1}}\xspace}
\newcommand{\tp}[1]{\ensuremath{\mathsf{#1}^\text{tp}}}
\newcommand{\eqdef}{\mathrel{\:\widehat{=}\:}}

\newcommand{\cs}[1]{\ensuremath{{{#1}_{{\scriptstyle{\mathtt{s}}}}}}}
\newcommand{\co}[1]{\ensuremath{{{#1}_{{\scriptstyle{\mathtt{o}}}}}}}

\newcommand{\ct}[1]{\ensuremath{{\hat{#1}}}}

\renewcommand{\max}[1]{\ensuremath{\mathsf{max}{#1}}}
\newcommand{\lastKey}[1]{\ensuremath{\max(\dom{#1})}}
\newcommand{\fresh}[1]{\ensuremath{\lastKey{#1}+1}}
\newcommand{\dom}[1]{\ensuremath{\mathsf{dom}({#1})}}
\newcommand{\Ux}[1]{\ensuremath{\mathbb{M}_{#1}}}

\newcommand{\HTj}[4]{\ensuremath{#1 : #2 #3 \,@\,#4}}
\def\GHTj#1[#2]#3#4#5{#1 : \mathsf{[}#2\mathsf{].}  \spec{#3} \\ & \spec{#4} \,@\,#5}
\def\GHTjaligned#1[#2]#3#4#5{#1 :~ \specK{\mathsf{[}#2\mathsf{].}}  & \spec{#3} \\ & \spec{#4} \,@\,#5}
\def\GHTjalignedNSP#1[#2]#3#4#5{#1 :~ \mathsf{[}#2\mathsf{].}  & #3 \\ & #4 \,@\,#5}

\def\HTjalignedline#1#2#3#4{#1 :~ & #2\ #3 \,@\,#4}
\def\GHTjalignedline#1[#2]#3#4#5{#1 :~ & \specK{\mathsf{[}#2\mathsf{].}} \spec{#3}\ \spec{#4} \,@\,#5}

\def\atomic<#1>{\langle #1 \rangle}
\newcommand{\set}[1]{\left\{ #1 \right\}}

\def\<#1,#2>{\langle #1 , #2 \rangle}

\def\codefont{\mathbf}
\def\DO{\codefont{do}}
\def\UNTIL{\codefont{until}}

\def\specP{\spec{P}}
\def\specQ{\spec{Q}}

\def\progLock{\prog{lock}}

\def\progUnlock{\prog{unlock}}

\def\rV{\res{V}}

\def\rTicket{\res{TL}}
\def\UTicket{\Ux{\rTicket}}
\def\STicket{\statesp{\rTicket}}

\def\currentT{\res{\color{green!55!black}{serve}}}
\def\servedT{\res{\color{white!55!black}{used}}}
\def\waitingT{\res{\color{yellow!45!red}{drawn}}}
\def\labelN{L}
\def\labelT{\set{\waitingT, \currentT, \servedT}}
\def\finmap{\rightharpoonup_\text{fin}}
\def\mapTctl{\Nat^+ \finmap \labelN    }

\newcommand{\ord}{\mathsf{ordered}}
\newcommand{\nogaps}{\mathsf{no\_gaps}}

\newcommand{\Hist}{\mathsf{Hist}}

\newcommand{\taketxTr}{\trans{taketx\_tr}}
\newcommand{\lockTr}{\trans{lock\_tr}}
\newcommand{\unlockTr}{\trans{unlock\_tr}}
\newcommand{\hmapsto}{\Mapsto}

\newcommand{\Nat}{\mathbb{N}}

\def\own{\mathsf{own}}
\def\nown{\overline{\mathsf{own}}}

\def\sigmaS{\cs{\sigma}}
\def\sigmaO{\co{\sigma}}
\def\sigmaT{\ct{\sigma}}

\newcommand{\sgS}[1]{\fapp{\cs{\sigma}}{#1}}
\newcommand{\sgO}[1]{\fapp{\co{\sigma}}{#1}}
\newcommand{\sgT}[1]{\fapp{\ct{\sigma}}{#1}}
\newcommand{\psiS}[1]{\fapp{\cs{\psi}}{#1}}
\newcommand{\psiO}[1]{\fapp{\co{\psi}}{#1}}
\newcommand{\psiT}[1]{\fapp{\ct{\psi}}{#1}}
\newcommand{\alphaS}[1]{\fapp{\cs{\alpha}}{#1}}

\newcommand{\phiS}[1]{\fapp{\cs{\phi}}{#1}}

\newcommand{\tauS}[1]{\fapp{\cs{\tau}}{#1}}
\newcommand{\tauO}[1]{\fapp{\co{\tau}}{#1}}
\newcommand{\tauT}[1]{\fapp{\ct{\tau}}{#1}}
\newcommand{\numCurrentX}{\#_\mathsf{\currentT}}
\newcommand{\numCurrent}[1]{\papp{\#_\mathsf{\currentT}}{#1}}
\newcommand{\funupdate}[3]{{#1}[{#2}\mapsto{#3}]}

\def\sgpsi{\sigma\!\otimes\!\psi}
\def\sgpsiS #1 {\cs{\sgpsi}(#1)}
\def\sgpsiO #1 {\co{\sgpsi}(#1)}
\def\sgpsiT #1 {\ct{\sgpsi}(#1)}

\def\sgpsimu{\sigma\!\otimes\!\psi\!\otimes\mu}
\def\sgpsimuS #1 {\cs{(\sgpsimu)}#1}
\def\sgpsimuO #1 {\co{\sgpsimu}#1}
\def\sgpsimuT #1 {\ct{\sgpsimu}#1}

\def\alphaS #1{\cs{\alpha}(#1)}

\newcommand{\rel}[1]{\mathrel{\bot_{#1}}}

\def\relAlpha #1 #2{#1 \rel{\alpha} #2}
\def\relChi #1 #2{#1 \rel{\chi} #2}
\def\relOmega #1 #2{#1 \rel{\omega} #2}

\newcommand{\ghostcode}[1]{\colorbox{gray!35}{\raisebox{0pt}[5pt][0pt]{\ensuremath{#1}}}}

\newcommand{\pcmA}{A}
\newcommand{\pcmB}{B}
\DeclareMathOperator{\join}{\bullet}
\def\+{\join}
\DeclareMathOperator{\undefOp}{\top}
\newcommand{\genrel}[2]{{#1}\mathrel{R}{#2}}
\newcommand{\orth}{\rel{}}
\newcommand{\unit}{\mathds{1}}
\newcommand{\valid}[2]{#1 \orth #2}
\newcommand{\comp}[2]{{#1}\circ{#2}}
\newcommand{\tepr}[2]{{#1}\otimes{#2}}
\DeclareMathOperator{\eql}{eql}

\usepackage{ifoddpage,marginnote}

\newcounter{ToDos}
\newcounter{WarnCounts}

\newcommand{\dt}[1]{\emph{\color{black}{#1}}} 

\newcommand{\acite}[3]{\ifappendix
  {Appendix~\ref{#1}}\else{\cite[Appendix~{#2}]{#3}}\fi}

\newcommand{\var}[1]{\ensuremath{\mathit{#1}}}

\newcommand{\bstar}{\,{\boldsymbol{*}}\,}

\newcommand{\ldot}{\mathord{.}\,}
\makeatletter
\def\incfetch{%
  \@ifnextchar({\incfetchapply}{\incfetchcore}
}
\def\incfetchcore{\textsf{inc\_and\_fetch}}
\def\incfetchapply(#1){\incfetchcore(#1)}
\makeatother

\newcommand{\separate}{separate\xspace}

\newcommand{\separateness}{separateness\xspace}

\newcommand{\separating}{separating\xspace}
\newcommand{\Separating}{Separating\xspace}

\newcommand{\SepRel}{\Separating Relation\xspace}
\newcommand{\seprel}{\separating relation\xspace}
\newcommand{\Seprels}{\Separating relations\xspace}
\newcommand{\SepRels}{\Separating Relations\xspace}
\newcommand{\seprels}{\separating relations\xspace}

\newcommand{\fapp}[2]{#1\,#2}
\newcommand{\papp}[2]{#1\,(#2)}

\newcommand{\thread}{\theta}

\newcommand{\ctr}{\var{tdr}}

\newcommand{\isalock}{\mathsf{is\_lock}}
\newcommand{\locked}{\mathsf{locked}}
\newcommand{\unlocked}{\mathsf{unlocked}}

\DeclareMathOperator{\id}{\iota}

\def\lockval{\mathtt{L}}
\def\unlockval{\mathtt{U}}

\newcommand{\Op}{\mathsf{Op}}

\newcommand{\cupdot}{\mathbin{\mathaccent\cdot\cup}}

\def\marknode<#1>[#2]#3%
{	
	\makebox[\widthof{#3}]{%
	\begin{tikzpicture}[remember picture,overlay]
		\node<#1>[text width=width("{#3}",
		text height=height("{#3}"),
		yshift=-0.5em,
		shape=ellipse,
		anchor=south west,
		ultra thick,draw=red] (#2) {};
	\end{tikzpicture}
	{#3}
}}

\begin{document}

\title[On Algebraic Abstractions for Concurrent Separation Logics]
      {On Algebraic Abstractions for Concurrent Separation Logics} 


\def\afIMDEA{
        \affiliation{\institution{IMDEA Software Institute}\country{Spain}}}
\def\afTEZOS{
        \affiliation{\institution{Nomadic Labs}\country{France}}}
\def\afUCM{
        \affiliation{\institution{Universidad Complutense de Madrid}\country{Spain}}}

\author{Franti\v{s}ek Farka}
\orcid{0000-0001-8177-1322}
\afIMDEA
\email{frantisek.farka@imdea.org}

\author{Aleksandar Nanevski}
\orcid{0000-0002-4851-1075}
\afIMDEA
\email{aleks.nanevski@imdea.org}

\author{Anindya Banerjee}
\orcid{0000-0001-9979-1292}
\afIMDEA
\email{anindya.banerjee@imdea.org}

\author{Germ\'{a}n Andr\'{e}s Delbianco}
\orcid{0000-0002-2249-1168}
\afTEZOS
\email{german@nomadic-labs.com}

\author{Ignacio F\'{a}bregas}
\orcid{0000-0002-3045-4180}
\afUCM
\email{fabregas@fdi.ucm.es}

\begin{abstract}
  Concurrent separation logic is distinguished by transfer of state
  ownership upon parallel composition and framing. The algebraic
  structure that underpins ownership transfer is that of partial
  commutative monoids (PCMs).  Extant research considers ownership
  transfer primarily from the logical perspective while comparatively
  less attention is drawn to the algebraic considerations.
  This paper provides an algebraic formalization of ownership transfer
  in concurrent separation logic by means of structure-preserving
  partial functions (i.e., morphisms) between PCMs, and an associated
  notion of \seprels.  Morphisms of structures are a standard concept
  in algebra and category theory, but haven't seen ubiquitous use in
  separation logic before. \Seprels are binary relations that
  generalize disjointness and characterize the inputs on which
  morphisms preserve structure.
  The two abstractions facilitate verification by enabling concise
  ways of writing specs, by providing abstract views of threads'
  states that are preserved under ownership transfer, and by enabling
  user-level construction of new PCMs out of existing ones.

\end{abstract}



\maketitle

\def\appUnlock{A}
\def\appStability{B}
\def\CAP{C}
\def\appCoq{D}


\section{Introduction}\label{s:intro}



The algebraic foundations of separation logic are rooted in the
discovery that the structure of \emph{partial commutative monoids
  (PCMs)} underpins the semantics of the key inference rules of
%
%
framing and parallel
composition~\cite{PymOY04,CalcagnoOY07,DinsdaleYoung-al:POPL13}. The
PCMs do so by mathematically representing the essential notions of
state ownership and ownership transfer, while abstracting the details
of the concrete memory models used by the programs.

%

In a nutshell, a PCM is a structure $(\pcmA, \join, \unit)$ on a
carrier set $\pcmA$, equipped with a (partial) binary operation
$\join$ (pronounced ``join''), which is commutative, associative, and
has $\unit$ as the unit.
The elements of the carrier $A$ model the \emph{private} state of
individual threads, and $\join$ models how the private states of two
children threads combine into the state of their parent. The operation
$\join$ is commutative and associative because the order of threads in
a thread pool is irrelevant for the computation. The operation
$\join$ is partial to signify that some state combinations are impossible. For
example, if $x \join y$ is undefined, then $x$ and $y$ can't be the
private states of two different concurrent threads,
\emph{simultaneously}.
%
The unit element represents the empty private state.  

The canonical PCM in separation logic is that of heaps, which are
finite maps from pointers (positive natural numbers) to values. The
$\join$ is the \emph{disjoint} union of heaps. It is undefined if the
operand heaps have a pointer in common, thus modeling that the private
heaps of two concurrent threads can't share pointers. The unit is the
heap with no pointers allocated.
When a parent forks two children threads, then its private heap is
divided disjointly among the children. Upon joining, the private,
disjoint heaps of the children are unioned to derive the heap of the
parent. This transfer of heap ownership between parent and children
threads is the defining pattern of separation logic.

While PCMs were originally used to explain the \emph{semantics} of
separation logic, more recent separation
logics~\cite{LeyWild-Nanevski:POPL13,Nanevski-al:ESOP14,Jensen-Birkedal:ESOP12,Appel-al:BOOK14,Jung-al:POPL15}
take a step further and employ PCMs in \emph{program specifications}
(henceforth: specs). In these logics the user may introduce various
PCMs to model custom notions of ghost state relevant to the
verification problem. Examples include
PCMs of permissions~\cite{Bornat-al:POPL05}, and PCMs of
histories~\cite{Sergey-al:ESOP15} for representing temporal (i.e.,
execution order) properties in the style of linearizability and other
consistency
criteria~\cite{del+ser+nan+ban:ecoop17,NanevskiBDF+oopsla19,sergey:oopsla16}. Having
arbitrary PCMs also facilitates the verification of graph
algorithms~\cite{Sergey-al:PLDI15}, which has been notoriously
difficult in heap-only separation logics.
%
%
These approaches therefore usefully combine the algebra of PCMs with
logical reasoning about state ownership and transfer.


In this paper, we take the PCM-based approach to specification
significantly further by introducing a theory of structure-preserving functions (\dt{morphisms}), and structure-preserving relations (\dt{\seprels}) on PCMs. Morphisms are \emph{partial},
as they preserve the PCM structure only on \emph{some} inputs. \Seprels are
binary relations that describe the inputs on which a morphism is
structure-preserving, and abstractly generalize heap disjointness.

The above development has two relevant consequences for
separation logic.
First, it immediately provides powerful user-level support for
constructing new PCMs out of existing ones. To see why such
construction is desirable, consider that to specify both spatial and
temporal properties of programs, the user may want to combine the PCMs
of heaps and histories into their Cartesian product, itself also a
PCM. But a standard use of morphisms in abstract algebra and category
theory is precisely in the definition of algebraic constructions,
where morphisms relate a construction to its components, e.g., how a
Cartesian product is associated with projection and pairing morphisms. We
illustrate this aspect of our contribution by introducing the
algebraic construction of a \emph{sub-PCM}, and showing how it applies
to verification in separation logic.

Second, the two concepts (morphisms and \seprels) provide ways to
\emph{abstract} from the concrete thread states; morphisms can
functionally \emph{compute} novel abstractions from a state, whereas \seprels
relate the states of a thread and its concurrent environment. Being
structure-preserving means that both respect the ownership transfer of
separation logic, as we shall see.
Together, the two concepts thus present: a novel foundation for
separation logic that facilitates systematic introduction of algebraic
concepts into specs; and a way to mathematically model the essentials
of a verification problem while abstracting from details of program
state.

\subsection{Morphisms as Ownership-Preserving
  Abstractions}\label{s:intromorph}
Glossing over the partiality of PCMs, to which we return in
Section~\ref{s:introseprel}, the standard algebraic definition says
that a morphism from the monoid $(\pcmA, \join_\pcmA, \unit_\pcmA)$
to the monoid $(\pcmB, \join_\pcmB, \unit_\pcmB)$ is a function
$\phi : \pcmA \to \pcmB$ that preserves the monoidal structure:
%
\begin{align}
\label{eq:mp1} \phi (\unit_\pcmA) ~ = ~ & \unit_\pcmB\\ 
\label{eq:mp2} \phi (x\join_\pcmA y) ~ = ~ & \phi(x) \join_\pcmB \phi(y)
\end{align}
We previously described $\join$ as a way to combine private states of
two children threads into the state of the parent. The above equations
then characterize $\phi$ as computing a view---an abstraction---of a
thread's private state, while preserving the thread-private nature of
the view.

To illustrate, consider how PCMs may model a mutually exclusive lock
that threads race to acquire.
We first require the PCM $O$ that formalizes lock ownership. $O$ has
the carrier $\{\own, \nown\}$, where $\own$ (resp. $\nown$) signifies
that the thread owns (resp. doesn't own) the lock.
%
%
The $\join$ computes the lock ownership of the parent thread from
those of the children by the following table, where $\own \join \own$
is undefined as two threads can't own the lock simultaneously,
and $\nown$ is the unit.
\[
\begin{array}{c c c c}
x & y & \quad & x \join y\\
\hline 
\own  & \own & & \mathsf{undefined}\\
\own & \nown & & \own \\
\nown & \own & & \own \\
\nown & \nown & & \nown 
\end{array}
\]
The table says that the lock is transferred from child to parent
upon joining, analogously to how the heap of a child is transferred to
the parent, as discussed before. If neither child owns the lock, then
the parent doesn't own the lock either.

A concrete implementation of the lock will typically require threads
to internally store much more private lock-related state than merely
an element of $O$. The extra state may be used for synchronization
purposes, or it may be ghost state required to formulate the logical
invariants of the locking algorithm, as often necessary for
verification. Let this private state be modeled by a PCM $X$. The
concrete definition of $X$ may differ between lock implementations and
proofs, but each should exhibit a function $\alpha : X \rightarrow O$
that computes the lock ownership status $\fapp{\alpha}{x}$ of a thread
from the thread's private state $x \in X$.

Moreover, $\alpha$ must be structure-preserving, and in particular
must satisfy equation~(\ref{eq:mp2}). To see what goes wrong if
$\alpha$ doesn't, suppose there're states $x$ and $y$ such that, e.g.,
$\fapp{\alpha}{x} = \fapp{\alpha}{y}=\nown$, but $\papp{\alpha}{x
  \join y} = \own$, to consider but one bad combination of values for
$\alpha$ (the other bad combinations are similarly absurd). Then we
have two children threads that don't own the lock, but their parent is
granted the lock upon joining, out of thin air. Such $\alpha$ violates
the transfer of lock ownership between children and parent threads,
and thus doesn't model locking.

\subsection{Partiality and \SepRels as Abstraction of
  Disjointness}\label{s:introseprel}

Taking into account that $\join_A$ may be undefined on some inputs,
it's clear that equation~(\ref{eq:mp2}) can't hold as stated, but must
be prefixed by some condition on $x$ and $y$. At the very least, such
condition should entail that $x \join_A y$ is defined, so that $\phi$
has an input value to which to apply, and on which $\phi$ itself is
defined.
More generally, we associate $\phi$ with a binary relation
$\rel{\phi}$ that captures when $\phi$ distributes over $\join$, via
the updated axiom
\begin{gather}
\label{eq:mp3}
\mbox{if}\ x \rel{\phi} y\ \mbox{then}\ x \join_\pcmA y\ \mbox{and}\ \phi (x \join_\pcmA y)\ \mbox{are defined and}\ \phi (x \join_\pcmA y) = \phi(x) \join_\pcmB \phi(y)
\end{gather}
The relation $\rel{\phi}$ will be a \emph{\seprel}, thus satisfying a
number of properties that we outline in Section~\ref{s:seprels}. One
of the properties that $x \rel{\phi} y$ entails is that $x \join_A
y$ is defined, or, equivalently, that $x$ and $y$ are
\emph{\separate\xspace} (denoted $x\,{\orth}\,y$). Clearly, this
notion generalizes disjointness of heaps and applies it to arbitrary
PCMs. Then a \seprel $\rel{\phi}$ represents a morphism-specific
notion of \separateness that strengthens the one inherited from the
underlying PCM.

Because \separateness determines when states of two threads combine
into a parent state, \seprels essentially provide a custom notion of
when two PCM elements can be considered as states of concurrent
threads, and thus also when a PCM element can be transferred from one
thread to another. A related important use of \seprels is in the
construction of sub-PCMs of the PCM $A$, whereby $\join_A$ is
restricted to the inputs admitted by the \seprel. These uses are
illustrated in Sections~\ref{s:examplesubpcm}
and~\ref{s:seprels}.


We also show in Section~\ref{s:seprels} that morphisms and \seprels
are closed under basic algebraic constructions.
%
%
For example, morphism kernels and equalizers are \seprels; restricting
a morphism by a \seprel produces a new morphism, etc. Thus, \seprels
are a natural algebraic structure to describe the inputs on which a
partial PCM morphism is structure-preserving (and defined).

\subsection{Morphisms and \SepRels in the Abstract}

We further consider how morphisms and \seprels interact to support
framing (or more generally, parallel composition) \emph{in the
  abstract}. In other words, if we have a spec involving morphisms and
\seprels whose exact definitions we want to hide, what properties must
be exposed 
to make it possible to frame the spec?
In Section~\ref{s:inversion}, we argue that what must be exposed is
that the morphisms and the \seprels respectively satisfy the novel property of
\dt{invertibility}, in addition to being structure-preserving
functions and relations. 
Framing in the abstract has been considered in related work on
concurrent abstract predicates (CAP)~\cite{DinsdaleYoung-al:ECOOP10}.
The novelty of our approach is the use of morphisms (i.e., functions)
rather than predicates (i.e., relations). When possible, functions are
preferred to relations, as results of functions needn't be named;
hence one can avoid existential quantification (e.g., consider
function vs.~relation composition). Section~\ref{s:relwork}
discusses
the relationship to concurrent abstract predicates.



\subsection{Use of Morphisms in Specs}
We show that morphisms allow the user to compute, directly in specs,
PCM values out of the state, without requiring almost any other
logical connectives familiar from separation logic. Thus, for the most
part, our specs won't use separating conjunction%
\footnote{Though we'll
  define a similar notion for use in proof outlines.} %
or separating
implication, or the numerous recent additions to separation logic
of \emph{assertions} in the form of modalities and custom notions of
implication~\cite{DinsdaleYoung-al:ECOOP10,Jung-al:POPL15,jung:jfp18,gra+biz+kre+bir:iron19}
and quantification~\cite{ArrozPincho-al:ECOOP14}.  Instead, we rely
only on standard constructs from higher-order logic to make and
combine statements about morphism values and \seprels. However, ours
is still a separation logic as we're concerned with PCMs and
ownership transfer.

As morphisms are just a special class of functions, they are
particularly well-suited to a formalization as a \emph{shallow
  embedding} in a system based on type theory such as Coq. We have
thus mechanized all the results from the paper by building on the
recent formulation of separation logic in Coq by
\citet{NanevskiBDF+oopsla19}. Morphisms and \seprels integrate very
naturally into this ambient theory, and don't require any particular
automation by tactics in order to be used effectively.
%
The resulting mechanization is available as a separate
artefact~\citep{supplemental}.

\section{PCM Abstractions by Example}
\label{s:example}

\subsection{Ticket Lock}\label{sec:tl}
To illustrate the issue at hand consider a simple synchronization
primitive, a ticket lock~\cite{mcs91,Lamport74a}. Ticket lock consists
of two shared pointers, the ticket dispenser $\ctr$, 
and the display $\var{dsp}$. The thread that wishes to acquire the
lock first increments $\ctr$ by the $\incfetch$
primitive.\footnotemark{} The thread then loops until the pointer
$\var{dsp}$ matches the value read from $\ctr$. The thread unlocks by
incrementing the value of $\var{dsp}$.
\begin{align*}
	\progLock \eqdef~ &
	  x \leftarrow \incfetch(\ctr);\\
          & \DO~~~y \leftarrow\,!\var{dsp}~~~\UNTIL\ x = y\\
	\progUnlock \eqdef~ &  \incfetch(\var{dsp})
\end{align*}
\noindent
Intuitively, the ticket lock's workflow resembles the ticket queue
management system that guides customers to a counter in a
bakery~\cite{Lamport74a}. Incrementing $\ctr$ corresponds to taking a
new ticket from the ticket dispenser, thus fixing a customer's
position in the queue. Looping corresponds to awaiting the ticket's
turn.
Incrementing $\var{dsp}$ signals, on the display, the next customer's
turn. The initial value of $\ctr$ is $0$; thus, the first ticket drawn
is $1$. The initial value of $\var{dsp}$ is $1$; thus, the first
thread that draws $1$ can immediately be served. The sequel continues
this analogy.

\footnotetext{\emph{Increment-and-fetch} is a generic
  RMW operation~\cite{Herlihy-Shavit:08} that atomically increments
  the value stored at $\ctr$ and returns the incremented
  value. Similar primitives exist in many systems,
  e.g.,~\textsf{\_\_atomic\_add\_fetch} of gcc.}

%
%
\subsubsection*{Specs.}
The specs of the two ticket lock programs should say that $\progLock$
acquires exclusive ownership of the ticket lock, and $\progUnlock$
releases it. We denote that by the following type
ascriptions.\footnote{For simplicity, we don't consider lock
  invariants that describe the heap that the lock protects. Attaching
  such invariants is an orthogonal issue to the topic of this paper
  and has been discussed in~\cite{NanevskiBDF+oopsla19}.\label{ftn:invars}}
\begin{eqnarray}
  \label{spec:lock}
 	\HTj{\progLock}
		{\spec{\lambda s \ldot \fapp{\cs{\alpha}}{s} = \nown
		}}
		{\spec{\lambda s \ldot \fapp{\cs{\alpha}}{s} = \own
		}}
		{\rTicket}\\
 \label{spec:unlock}\HTj{\progUnlock}
		{\spec{\lambda s \ldot \fapp{\cs{\alpha}}{s} = \own
		}}
		{\spec{\lambda s \ldot \fapp{\cs{\alpha}}{s} = \nown
		}}
		{\rTicket}
\end{eqnarray}
Unlike most separation logics, we make the binding of the state $s$ in
the assertions explicit by means of $\lambda$, as customary in
higher-order logic.
In the above specs, $\alpha$ is a morphism from the underlying PCM of
the state $s$, whose exact definition we want to keep abstract, to the
PCM $O$ from Section~\ref{s:intromorph}.
Several questions arise. Although a client can reason with the specs,
they appear too abstract: how can the specs be established in the
first place?  After all, on inspection of the implementations of
$\progLock$ and $\progUnlock$ above, it isn't obvious how morphism
$\alpha$ is even involved. It turns out that we will require concrete
specs of the implementations and then hide implementation-level
details to define $\alpha$ and obtain the abstract
specs~(\ref{spec:lock}) and (\ref{spec:unlock}).
%
%
But then how do morphisms and separating relations interact with the
\emph{concrete} specs? How do they work with framing of the concrete
specs? How do they work under abstraction? The sequel answers these
questions after first introducing the basics of our type-theoretic
approach.

\subsection{Hoare Types, States and Specifications}
\subsubsection*{Hoare Types}
A Hoare type~\cite{Nanevski-al:ICFP06,NanevskiBDF+oopsla19} is a
dependently typed state and concurrency (and divergence) monad,
indexed with 
a spec in the style of separation logic.
Concretely, in the judgment $\HTj{e}{\specP}{\specQ}{\rV}$, $P$ and
$Q$, both predicates over state $s$, are respectively the pre- and
postcondition of a program $e$, in the sense of partial correctness.
%
%
$\rV$ is a \emph{resource}, i.e., a state transition system describing
the atomic state changes that $e$ is permitted. Two programs can be
safely composed, sequentially or in parallel, only if they are typed
by the same resource. The resource thus serves as a bound on the
interference that concurrent threads can perform on each other's
executions, enabling a form of \emph{rely-guarantee}
reasoning~\cite{Jones:TOPLAS83}. As $\progLock$ and $\progUnlock$
share the resource $\rTicket$ (to be defined soon), they can be
composed.


\subsubsection*{States}
In our ambient type theory~\cite{NanevskiBDF+oopsla19}, states are
\emph{subjective}~\cite{LeyWild-Nanevski:POPL13}. That is, each state
$s$ is a pair $(\cs{s}, \co{s})$, where $\cs{s}$ and $\co{s}$ are
referred to as \emph{self} and \emph{other} components,
respectively. The $\cs{s}$ component describes the private state of a
thread, whereas $\co{s}$ describes the combined state of all the other
threads, that is, the concurrent environment.\footnote{States
  in~\cite{NanevskiBDF+oopsla19} also contain the third component
  $s_j$ describing shared state, but we won't need it here.} 
Thus, \emph{self} and \emph{other} components model, \emph{at the
  level of state}, the same dichotomy modeled by the rely and
guarantee \emph{transitions} of rely-guarantee reasoning.
%
%
%
%
The value of $\co{s}$ may be used in specs, but a program can't alter
it. Both $\cs{s}$ and $\co{s}$ are elements of one and the same PCM.
When we access state components by a morphism $\phi$, we attach the
subscript to the morphism and write $\fapp{\cs{\phi}}{s}$ and
$\fapp{\co{\phi}}{s}$ instead of $\fapp{\phi}{\cs{s}}$ and
$\fapp{\phi}{\co{s}}$ respectively. We write $\fapp{\hat\phi}{s}$ for
$\fapp{\cs{\phi}}{s} \join \fapp{\co{\phi}}{s}$.
We also implicitly assume that in every state $s$, the components
$\cs{s}$ and $\co{s}$ are \emph{\separate}; that is,
$\cs{s} \join \co{s}$ is defined in the PCM of the resource. Using the
notation from Section~\ref{s:introseprel}, this is denoted as
$\cs{s} \orth \co{s}$.

%
\begin{figure}
\[
\begin{array}{c@{\!\!\!\!}c@{\!\!\!\!}c}
\begin{array}{c}
\begin{tikzpicture}[scale=0.94, transform shape]
  \filldraw[fill=gray!70, draw=black] (0, 0) -- (-30:1cm) arc (-30:90:1cm) -- cycle; 
  \filldraw[fill=gray!20, draw=black] (0, 0) -- (90:1cm) arc (90:210:1cm) -- cycle; 
  \filldraw[fill=gray!70, draw=black] (0, 0) -- (210:1cm) arc (210:330:1cm) -- cycle; 
  \filldraw[fill=white, draw=black] (0,0) circle (0.53cm);
  \node[scale=0.9] at (155:0.75cm) {$a_1$};
  \node[scale=0.9] at (-90:0.75cm) {$a_2$};
  \node[scale=0.9] at (25:0.75cm) {$a_3$};
  \newlength\stxtwdl%
  \settowidth{\stxtwdl}{$\cs{s} = a_1 \hspace{3mm}$}%
  \node[scale=0.8, anchor=east] at (225:1cm) {$\cs{s_1} = a_1\quad$};
  \node[scale=0.8, anchor=west] at (45:1cm) {$\quad \co{s_1} = a_2 \join a_3$};
\end{tikzpicture}\\
{\fontsize{9}{9}\textrm{(1) Thread $\thread_1$}}
\end{array}
&
\begin{array}{c}
\begin{tikzpicture}[scale=0.94, transform shape]
  \filldraw[fill=gray!70, draw=black] (0, 0) -- (-30:1cm) arc (-30:90:1cm) -- cycle; 
  \filldraw[fill=gray!70, draw=black] (0, 0) -- (90:1cm) arc (90:210:1cm) -- cycle; 
  \filldraw[fill=gray!20, draw=black] (0, 0) -- (210:1cm) arc (210:330:1cm) -- cycle; 
  \filldraw[fill=white, draw=black] (0,0) circle (0.53cm);
  \node[scale=0.9] at (155:0.75cm) {$a_1$};
  \node[scale=0.9] at (-90:0.75cm) {$a_2$};
  \node[scale=0.9] at (25:0.75cm) {$a_3$};
  \newlength\stxtwdr%
  \settowidth{\stxtwdr}{$\cs{s} = a_1 \hspace{3mm}$}%
  \node[scale=0.8, anchor=east] at (225:1cm) {$\cs{s_2} = a_2\quad$};
  \node[scale=0.8, anchor=west] at (45:1cm) {$\quad \co{s_2} = a_3 \join a_1$};
\end{tikzpicture}\\
{\fontsize{9}{9}\textrm{(2) Thread $\thread_2$}}
\end{array}
&
\begin{array}{c}
\begin{tikzpicture}[scale=0.94, transform shape]
  \filldraw[fill=gray!70, draw=black] (0, 0) -- (-30:1cm) arc (-30:90:1cm) -- cycle; 
  \filldraw[fill=gray!20, draw=black] (0, 0) -- (90:1cm) arc (90:210:1cm) -- cycle; 
  \filldraw[fill=gray!20, draw=black] (0, 0) -- (210:1cm) arc (210:330:1cm) -- cycle; 
  \filldraw[fill=white, draw=black] (0,0) circle (0.53cm);
  \node[scale=0.9] at (155:0.75cm) {$a_1$};
  \node[scale=0.9] at (-90:0.75cm) {$a_2$};
  \node[scale=0.9] at (25:0.75cm) {$a_3$};
  \newlength\stxtwd%
  \settowidth{\stxtwd}{$\cs{s} = a_1 \join a_2 \hspace{3mm}$}%
  \node[scale=0.8, anchor=east] at (-225:1cm) {\makebox[\stxtwd][l]{$s = s_1 \star s_2$:}};
  \node[scale=0.8, anchor=east] at (225:1cm) {$\cs{s} = a_1 \join a_2 \quad$};
  \node[scale=0.8, anchor=west] at (45:1cm) {$\quad \co{s} = a_3$}; 
\end{tikzpicture}\\
{\fontsize{9}{9}\textrm{(3) Parent thread $\thread = \thread_1 \parallel \thread_2$}}
\end{array}
\end{array}\vspace{-1em}
\]
\caption{States of concurrent threads. \emph{Self} components are in 
  light shade, \emph{other} components are in dark. Adapted from 
  \cite{NanevskiBDF+oopsla19}.}\label{fig:subjectivity}\vspace{-3mm}
\end{figure}
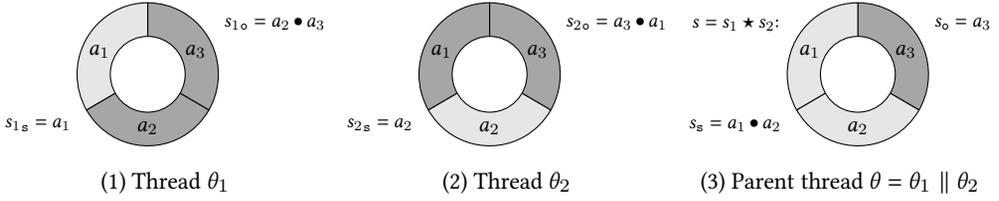

Figure~\ref{fig:subjectivity} illustrates the interaction among the
state components of concurrent threads. Consider three threads,
$\thread_1$, $\thread_2$, and $\thread_3$, running concurrently, and
without any additional threads. Their respective states must have the
forms $s_1 = (a_1, a_2 \join a_3)$, $s_2 = (a_2, a_3 \join a_1)$, and
$s_3 = (a_3, a_1 \join a_2)$, because any of the two threads combined
represent the concurrent environment of the third. Thus, the join of
the \emph{self}'s of any two threads must equal the \emph{other} of
the third.
If $\thread$ is the parent thread of $\thread_1$ and $\thread_2$, then
its state is $s = (a_1 \join a_2, a_3)$, since $\thread$ is the
combination of $\thread_1$ and $\thread_2$, and has $\thread_3$ as its
environment. In particular, the join of the \emph{self} and
\emph{other} components is invariant across all the
threads. Figure~\ref{fig:subjectivity} illustrates these relations.
Moreover, we abbreviate the relationship of the state $s$ of the parent
thread $\thread$ and the states $s_1$ and $s_2$ of children threads
$\thread_1$ and $\thread_2$ by $s = s_1 \star s_2$.

\subsubsection*{Morphisms and \SepRels}
%
In the types of $\progLock$ and $\progUnlock$, $\alpha$ computes the
lock ownership information from $\cs{s}$. It's therefore apparent that
the types capture what's desired: that the $\progLock$ program starts
not owning the lock (precondition $\fapp{\cs{\alpha}}{s} = \nown$),
and acquires the lock upon termination (postcondition
$\fapp{\cs{\alpha}}{s} = \own$), and conversely for
$\progUnlock$. We'll see examples of other morphisms and \seprels
shortly, when we discuss the internal definition of the state.



\subsection{Internal State of the Ticket Lock}
\label{s:intstate}
Recall that our goal is to define morphism $\alpha$ and reach the
abstract specs~(\ref{spec:lock}) and (\ref{spec:unlock}) via concrete
specs of the implementations of $\progLock$ and $\progUnlock$.  To
that end, we next design the ghost state of the ticket lock so that we
can express the internal logical invariants needed for the typing
derivations of the implementations of $\progLock$ and $\progUnlock$.
Later, the morphism $\alpha$ will abstract these internals to an
element of $O$.
We use the following PCM $U$ for the internals.
\begin{align}
	\label{eq:repr}
	U =~& \mapTctl \qquad \mbox{where $\labelN = \labelT$}
\end{align} 
%
%
Here, $\mapTctl$ is the type of finite (partial) maps from positive
natural numbers, representing tickets. Given a ticket $t$, the value
of the map at $t$ is one of the \emph{labels} in the set $\labelN$,
denoting the status of the ticket according to the ticket lock
workflow from Section~\ref{sec:tl}: $\waitingT$ means that $t$ has
been drawn from the dispenser and the thread holding $t$ is waiting to
be called on the display; $\currentT$ means that $t$ has been called
on the display and the thread has begun its turn holding the lock; and
$\servedT$ means $t$'s turn at the counter
has finished, and the thread holding $t$ has unlocked by signaling
$t+1$'s turn on the display. Notice that we don't throw away tickets,
but just change their status in the map to reflect their progress
through the bakery. The map thus serves as a form of history of the
bakery.
Similarly to heaps, the type $\mapTctl$ is a PCM under the operation
of disjoint union $\cupdot$ of maps, which is undefined if the two
operands share a ticket. The unit is the empty (i.e., nowhere defined)
map $\emptyset$. We take $\Nat^+$ as the domain instead of $\Nat$ in
order to exclude the ticket $0$, as the latter can't be drawn from the
dispenser.
%

Given a ticket map $x \in U$ that represents the history of tickets in
the bakery, we can compute out of $x$ the ticket called on the display
by the following definition, where we assume that $\mathsf{max}$ of
the empty set of natural numbers is by default the value $0$.
\[
\fapp{\psi}{x} = \max{\{t \in \dom{x} \mid \fapp{x}{t} = \servedT\}} + 1
\]
Indeed, according to the workflow of ticket locks, upon finishing its
turn with the lock a thread holding the ticket $t$ sets the display to
$t+1$ to call the next thread in the queue. Thus, the value of the
display, which at that point denotes the currently called ticket, is
one larger than the largest $\servedT$ ticket in $x$.
The function $\psi : U \rightarrow \Nat^+$ is our first example of a
morphism, where we endow $\Nat^+$ with the PCM structure
$(\Nat^+, \mathsf{max}, 1)$.
%
%
Indeed,
it's easy to see that for any two disjoint ticket maps $x \orth y$, we
have
\[
\fapp{\psi}(x \join y) = \mathsf{max}(\fapp{\psi}{x}, \fapp{\psi}{y})
\]
Moreover, $\fapp{\psi}{\emptyset} = 1$, and $1$, being the smallest
element of $\Nat^+$, is the unit w.r.t.  $\mathsf{max}$.



\subsubsection*{Morphism Notation for Ticket Locks}
Let us name the identity morphism on $U$ as $\sigma$. Giving a special
name to the identity morphism will provide for uniform notation in our
specs, where we apply $\sigma$, $\psi$ and other morphisms to compute
various values from states.
In particular, when applying morphisms $\sigma, \psi$ to
state 
$s = (\cs{s},\co{s})$, and according to the morphism notation from the
previous section, we use the following expressions to denote various
ticket maps and values.
\begin{itemize}
\item $\sgS{s}$ denotes the \emph{self} map of tickets. These
  are the tickets, and their status, that the thread under
  consideration (henceforth ``we'' or ``us'') has drawn from the
  dispenser.
\item $\sgO{s}$ denotes the \emph{other} ticket map. These are the
  tickets, and their status, that every other thread but ``us''
  (henceforth ``others'') has drawn from the dispenser.
\item $\psiS{s}$ denotes the \emph{self} value of the last \emph{called}
  ticket. This is the ticket that ``we'' have called by incrementing
  the display upon finishing our last turn at the counter, to call the
  next thread in the ticket queue.
\item $\psiO{s}$ is the ticket last called by ``others'', when they
  finished their turns at the counter.
\end{itemize}

The combined ticket map $\sgT{s} = \sgS{s} \join \sgO{s}$ and the
value
$\psiT{s} = \psiS{s} \join \psiO{s} = \mathsf{max}(\psiS{s},
\psiO{s})$ have further important meanings.
%
As tickets are drawn in order, we can compute the
current value of the ticket dispenser pointer $\ctr$ as
$\lastKey{\sgT{s}}$. Similarly, we can compute the value of the display
pointer $\var{dsp}$ as $\psiT{s}$.
Therefore, our specs 
needn't explicitly store the values of $\ctr$ and
$\var{dsp}$, or any other shared state. \emph{In specs}, any shared
state can generally be 
computed out of \emph{self} and 
\emph{other} ghost components that suitably track the history of the
updates to that shared state, just like $\sgT$ and $\psiT$ compute the
values of $\ctr$ and $\var{dsp}$ out of the \emph{self} and \emph{other}
ticket maps.\footnote{Of course,
  one needs to relate the ghost to concrete program state, shared or
  private, but that's beyond our scope here. We refer
  to~\cite{NanevskiBDF+oopsla19} for more details on how this
  relationship is made in the ambient theory.}

\subsection{Concrete Specs, Ghost Code, and Proof
  Outlines}\label{s:transitions}
With the internal state defined, we can next establish the following
types for the implementations, in Section~\ref{sec:tl}, of $\progLock$ and
$\progUnlock$. The types are \emph{concrete}, because they
specify $\progLock$ and $\progUnlock$ in terms of components of the
underlying PCM $U$ using morphisms $\sigma$ and $\psi$, thus exposing
the internal state of ticket lock. In the example, we denote by
$t \hmapsto l$ the singleton map that assigns label $l$ to a ticket
$t$, and is undefined elsewhere.
\begin{align*}
\HTjalignedline{\progLock}
	{\spec{\lambda s\ldot
		\sgS{s} = \emptyset 
		}} 
	{\spec{\lambda s\ldot 
		\sgS{s} = (\psiT{s}) \hmapsto \currentT 
		}}
	{\rTicket}\\
\GHTjalignedline{\progUnlock}[t]
	{\lambda s\ldot
		\sgS{s} = t \hmapsto \currentT 
                \land t = \psiT{s} 
		}
	{\lambda s\ldot
		\sgS{s} = t \hmapsto \servedT 
		}
	{\rTicket}
\end{align*}

The spec for $\progLock$ says that initially the ghost ticket map is
empty. Thus, as customary in separation logic, it can be framed to any
ticket map.
Upon termination, we hold the ticket being displayed and this ticket
is labeled as $\currentT$ in our map
($\sgS{s} = (\psiT{s}) \hmapsto \currentT$). Notice that the value
$\psiT{s}$ in the postcondition is \emph{stable} under interference, as other
threads can't change the display because we hold the lock when we're being
served. In
particular, they can't change $\psiO{s}$ which is a factor in the
computation of $\psiT{s} = \psiS{s} \join \psiO{s}$.

The spec for $\progUnlock$ says that we hold the displayed ticket $t$
($\sgS{s} = t \hmapsto \currentT$ and $t = \psiT{s}$). Upon
termination, we still hold $t$, but it's now labeled as $\servedT$, to
indicate we finished our turn.
The Hoare type for $\progUnlock$ explicitly binds the variable $t$,
denoted by $[t]$, to snapshot the initial value of the display and to
allow its use both in the precondition and the postcondition. The
scope of $t$ extends through the precondition and postcondition to the
right of the binding $[t]$.
We couldn't have ascribed to $\progUnlock$ the postcondition
$\lambda s\ldot\sgS{s} = \psiT{s} \hmapsto \servedT$  
because the value $\psiT{s}$ in the postcondition isn't
stable.
Indeed, after we unlock, other threads can get their turn at the
counter and increment the display. Thus, we use $t$ to explicitly bind
the stable value that the display has when $\progUnlock$ is invoked,
and we hold the lock.

%


We emphasize how morphisms in the above specs combine in the standard
mathematical fashion to compute various required values. For example,
we apply $\psi$ to $\cs{s}$ and $\co{s}$ to obtain $\psiS{s}$ and
$\psiO{s}$, and then combine the two into the expression $\psiT{s} =
\psiS{s}\join\psiO{s}$, to define $\sgS{s}$. But for this to be
possible, we had to make the binding of the state $s$ \emph{explicit}
in the assertions of the Hoare triple, so that $\psi$ and $\sigma$
could be applied to the different projections of the same state
$s$. Had we kept $s$ \emph{implicit}, as customary in separation logic
assertions, expressing the above specs would have required somewhat
more logical machinery.
%
%
This convenience afforded by morphisms and explicit states extends to
proof outlines, and to the definitions of resource transitions (see
below) which relate two states, the input and output states, that are
usefully differentiated by the explicit naming.

\subsubsection*{Transitions and State Space of the Resource $\rTicket$}
Before we can derive the types for $\progLock$ and $\progUnlock$, we
need to annotate the programs with \emph{ghost code}, i.e., code that
manipulates the ghost state expressed in terms of $\sigma$ and $\psi$.
In our ambient type theory, the ghost code is formed by transitions of
the resource (i.e., the state-transition system) of the specs; in the
current example, transitions of the resource $\rTicket$.
In Figure~\ref{fig:ststicket} we show the three transitions that define
$\rTicket$:
\taketxTr, \lockTr, and \unlockTr.
Each is a
relation over the initial state $s$ and final state $s'$, and defines
one of the three basic changes that ticket lock programs can perform
over the state.  
%
We denote by $\funupdate{f}{x}{a}$ the function obtained by changing
the value of function $f$ at point $x$ to the value $a$.

\begin{figure}
\begin{subfigure}{\textwidth}
\begin{align*}
\taketxTr~s~s' \eqdef &~
  \sgS{s'} = (t \hmapsto \waitingT) \cupdot  \sgS{s}
	&& \mbox{where $t = \fresh{\sgT{s}}$}\\
\lockTr~s~s' \eqdef  &~
  \papp{(\sgS{s})}{t} = \waitingT 
  	\land 
  \sgS{s'} = \funupdate{(\sgS{s})}{t}{\currentT}
       && \mbox{where $t = \psiT{s}$} \\
\unlockTr~s~s' \eqdef &~
  \papp{(\sgS{s})}{t} = \currentT \land 
  \sgS{s'} = \funupdate{(\sgS{s})}{t}{\servedT}\
   && \mbox{where $t = \psiT{s}$}
\end{align*}
\end{subfigure}

\vspace{1mm}

\begin{subfigure}{0.60\textwidth}
\begin{flalign*}
   &s \in \STicket \eqdef \papp{\ord}{\sgT{s}}
   \wedge \papp{\nogaps}{\sgT{s}}&
\end{flalign*}
where
\begin{align*}
\papp{\ord}{x} \eqdef &
	\set{t \mid \papp{x}{t} = \servedT} <
	\set{t \mid \papp{x}{t} = \currentT} \wedge \\ 
 	& \set{t \mid \papp{x}{t} = \currentT} < \set{t \mid 
    \papp{x}{t} = \waitingT} \wedge\\
 	& \set{t \mid \papp{x}{t} = \servedT} < \set{t \mid 
    \papp{x}{t} = \waitingT}\\
\papp{\nogaps}{x} \eqdef & 
	\forall t \in \Nat^+\ldot t+1 \in \dom{x} \Rightarrow t \in \dom{x}
\end{align*}
\end{subfigure}
\begin{subfigure}{0.29\textwidth}
\begin{tikzpicture}
\centering
\node[circle,draw=black] (D) at (0,0) {\STicket};
\draw (D) edge[loop above] node {\unlockTr} (D);
\draw (D) edge[loop right] node {\lockTr} (D);
\draw (D) edge[loop below] node {\taketxTr} (D);
\end{tikzpicture}
\end{subfigure}
\caption{The state transition system \rTicket. For sets $S, T$, the
  notation $S < T$ means $\forall s\in S, t\in T.\, s < t$.}\vspace{-3mm}
\label{fig:ststicket}
\end{figure}
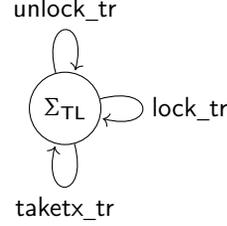


In the transition $\taketxTr$, the smallest undrawn ticket in the
state $s$, $\fresh{\sgT{s}}$, is added into the \emph{self} component
$\sgS{s'}$ and labeled $\waitingT$. Thus $\taketxTr$ models a thread
drawing a fresh ticket.
In the transition $\lockTr$, the value $t = \psiT{s}$ is the ticket on
display. This transition updates the ticket map $\sigma$ at $t$ from
$\waitingT$ to $\currentT$ to model that the thread noticed its ticket
called on the display, and took its turn at the counter.
The $\unlockTr$ transition checks that the ticket $t$ being displayed
is owned by the thread and is being served: so $(\sgS{s})(t) =
\currentT$.
%
%
The transition updates the status of $t$ to $\servedT$ to model
finishing the turn. Note that from the definition of $\psi$
(Section~\ref{s:intstate}), one can immediately compute that
$\psiS{s'} = t+1$ because $t$ is the largest used ticket in $s'$.

We emphasize that resource transitions aren't arbitrary relations on
states. Rather, as customary in separation logic~\cite{CalcagnoOY07},
they must satisfy the important property of \emph{locality}. The
latter constrains the behavior of a transition under ownership
transfer, and is necessary for the soundness of the rules of frame and
parallel composition. The precise definition of locality in the
subjective setting is given by the ambient type
theory~\cite[Definition 3.5]{NanevskiBDF+oopsla19}. Here, we just
mention that the three transitions of $\rTicket$ are all local, which
we proved in the Coq code. Because locality constrains ownership
transfer, these proofs essentially rely on the functions $\sigma$ and
$\psi$ being morphisms. In Section~\ref{s:frameandmorph} we illustrate
how morphisms behave under ownership transfer and specifically under
the rule of frame.

$\rTicket$, being a state transition system, requires a state space in
addition to transitions.  The state space $\STicket$, given in
Figure~\ref{fig:ststicket}, is a subset of $U \times U$ that the
transitions preserve.
Thus, the state space imposes natural properties of ticket locks that:
(1) tickets go through the bakery in order, i.e. all $\servedT$
tickets are smaller than $\currentT$ tickets, which in turn are
smaller than all $\waitingT$ tickets, as defined by the predicate
$\ord$ in Figure~\ref{fig:ststicket}; and (2) tickets are drawn
consecutively from the dispenser and none are skipped, as defined by
the predicate $\nogaps$ in Figure~\ref{fig:ststicket}.
%



\subsubsection*{Ghost Code Annotation}
We elide the discussion on how to formally factor transitions into the
ghost code, and refer to the ambient type
theory~\cite{NanevskiBDF+oopsla19} for details. Instead, we decorate
$\progLock$ and $\progUnlock$ below to \emph{informally} illustrate
when the various transitions are invoked to change the ghost
components of the state.
%
%
\[\begin{aligned}
 \progLock \eqdef~ &
		x \leftarrow \atomic<\incfetch(\ctr);
		\ghostcode{
			\taketxTr
		} >; \\
                & \DO~~~y \leftarrow !\var{dsp}~~~\UNTIL\ x = y;\\
		& \atomic<\ghostcode{\lockTr}> \\
\progUnlock \eqdef~ & \left\langle \incfetch(\var{dsp}) ; \ghostcode{
			\unlockTr	
		}\right\rangle
\end{aligned}\]
In the above code, angle brackets $\langle -\rangle$ signify that
the code they enclose executes atomically, that is without interference from
other threads.
In the first
case, the value returned from the agglomeration of actual with ghost
code is the value returned by the actual code itself. For example, the
$\progLock$ program executes $\taketxTr$ atomically with 
%
%
the call to $\incfetch{(\ctr)}$, to bind to $x$ the incremented value
of $\ctr$, and set the status of $x$ in the ghost state to
$\waitingT$.
When the condition $x = y$ is satisfied, since $y$ is assigned the
value of $\var{dsp}$, the ticket $x$ is called on the display. The
$\progLock$ program then executes $\lockTr$ as its final command to
set the status of ticket $x$ to $\currentT$. This models taking the
turn at the counter and completes the acquisition of the lock.
Similarly, $\progUnlock$ executes $\unlockTr$ to record in the ghost
state that the display is incremented upon unlocking.

\subsubsection*{Proof Outlines}\label{s:spec}
We next present the proof outline for $\progLock$ and discuss its key
points.
{
\newcounter{lockcodea}
\def\lineno{\stepcounter{lockcodea}\textsc{\thelockcodea}}
\[\begin{array}[t]{r@{\quad}l}
%
\lineno.& \proofspec{
	\sgS{s} = \emptyset 
}\\
\lineno.& x \leftarrow \atomic<\incfetch(\ctr) ;
                 \ghostcode{\taketxTr} >\\
\lineno.&     \proofspec{
	\sgS{s} = x \hmapsto \waitingT 
        \land \psiT{s} \leq x 
	}\\
\lineno.& \DO~~~y \leftarrow\,!\var{dsp}\\
\lineno.&     \proofspec{
	\sgS{s} = x \hmapsto \waitingT 
        \land y \leq \psiT{s} \leq x 
	}\\
\lineno.& \UNTIL\ x = y; \\
\lineno.& \proofspec{
	\sgS{s} = x \hmapsto \waitingT 
        \land y = \psiT{s} = x 
	}\\
\lineno.& \atomic< \ghostcode{\lockTr} >;\\
\lineno.& \proofspec{
	\sgS{s} = \psiT{s} \hmapsto \currentT 
	}\\
\end{array}\]
}

%
%
Line 1 is the precondition for $\progLock$. 
%
%
Line 3 shows that after the execution of $\taketxTr$, the drawn ticket
$x$ is the (only) ticket in $\sgS{s}$. Moreover, $x$ is computed by
$\incfetch$, and hence is one larger than the last ticket drawn. More
precisely, $x$ is bound to $\lastKey{\sgT{s}}$, for $s$ taken at line 2.
%
Now, from the definition of $\psi$, it must be that $\psiT{s} \leq x$
at the state $s$ taken at line 2. Indeed, $\psi$ computes the largest
$\servedT$ ticket, and $x$ equals the largest ticket, $\servedT$ or
not.
The property $\psiT{s} \leq x$ propagates to line 3 and beyond because
it's stable under interference. Other threads can execute the
transitions of $\rTicket$ over their own states to increase the
display (by increasing $\psiO{s}$ and thus also $\psiT{s}$), but can't
increase the display beyond
$x$. In~\acite{s:stability}{\appStability}{POPL21arxiv} we formally
establish this stability property.
%
%
For the do-until loop (lines 4-6), the loop invariant is on line 5: it
conjoins the property that $y$ is smaller than the displayed ticket
($y \leq \psiT{s}$). This property holds in the loop because line 4
stores the display value into $y$, after which the display may be
further incremented by other threads.
Line 7 marks the exit from the loop, thus the loop invariant holds
together with the condition $x = y$ for exiting the loop. This
immediately gives that $x = \psiT{s}$, which is a precondition for $\lockTr$.
%
Finally, line 9 directly follows
from line 7 by the definition of $\lockTr$.

\subsection{Framing and Morphisms}\label{s:frameandmorph}
The above spec for $\progLock$ is in the small footprint
style, where the spec's precondition uses $\emptyset$ for
$\sgS{s}$. A natural question is how this spec--which employs morphisms--can
be lifted to large footprints. In other words, how do we employ the frame rule
by using, as a frame, an arbitrary ticket map $k$ for $\sgS{s}$ in the precondition? Framing is a standard operation in separation logic, but works somewhat
differently in the setting with \emph{self} and \emph{other}
variables, and in the presence of morphisms. 
%
%
\newcounter{lockcodealphathree}
\def\lineno{\stepcounter{lockcodealphathree}\textsc{\thelockcodealphathree}}
\[\begin{array}[t]{r@{\quad}l}
\lineno.& \proofspec{
              \sgS{s} = k 
		}\nonumber\\
\lineno.& \proofspecL{
       \exists s_1\ s_2\ldot \!\!\!\begin{array}[t]{l}
            s = s_1 \star \frameC{s_2} \land 
            \sgS{s_1} = \emptyset 
            \land 
            \frameC{\sgS{s_2} = k 
           }\}
           \end{array}
	}\nonumber\\
\lineno.&	\proofspecL{\!\!\!\begin{array}[t]{l}
                ((\lambda s\ldot\sgS{s} = \emptyset 
                ) \bstar 
		\frameC{(\lambda s\ldot\sgS{s} = k 
		)})(s)\}
                \end{array}}\nonumber\\ 
\lineno.& {\progLock}\\
\lineno.&	\proofspecL{\!\!\!\begin{array}[t]{l}
                ((\lambda s\ldot\sgS{s} = \psiT{s} \hmapsto \currentT 
                ) \bstar 
                \frameC{
		(\lambda s\ldot\sgS{s} = k 
                )})(s)\}
                \end{array}}\nonumber\\
\lineno.& \proofspecL{\exists s_1\ s_2\ldot \!\!\!\begin{array}[t]{l}
               s = s_1 \star \frameC{s_2} \land 
		\sgS{s_1} = \psiT{s_1} \hmapsto \currentT 
                \land 
                \frameC{
                \sgS{s_2} = k 
		} \}
                \end{array}}\nonumber\\
\lineno.&  \proofspec{\sgS{s} = (\psiT{s} \hmapsto \currentT) \cupdot k
            }\nonumber
%
%
\end{array}\]

%
On line 1, we start with $\sgS{s}=k$; thus
$s=(k,\co{s})$.
Line 2 expands line 1 into a form suitable for applying the frame
rule. It posits that $s$ can be split into states $s_1$ and $s_2$ such
as $s = s_1 \star s_2$. It's easy to see that this holds: we can
represent $\cs{s} = \emptyset \join k$, thus pick
$s_1 = (\emptyset, k \join \co{s})$ and
$s_2 = (k, \emptyset \join \co{s})$ (see
Figure~\ref{fig:subjectivity}).
%
%
Line 3 represents line 2 using separating conjunction, which isn't a
primitive of our logic, but is defined in the ambient theory in the
customary way, modulo the use of subjective state splitting
(Figure~\ref{fig:subjectivity}):
\[
  P \bstar Q = \lambda s\ldot \exists s_1\ s_2\ldot s = s_1 \star s_2 
  \land \fapp{P}{s_1} \land \fapp{Q}{s_2}
\]
Line 5 applies the frame rule to the intermediate spec for $\progLock$
and the frame $(\lambda s\ldot\sgS{s} = k)$
(given in $\frameC{\mbox{color}}$ above). Line 6 unfolds the
definition of separating conjunction,
and line 7 collapses line 6, relying on the following two critical
points.
%

The first critical point is that $\psiT{s}=\psiT{s_1}$.  Indeed,
\[
\begin{array}{llll}
\psiT{s} &=&\psiS{s} \join \psiO{s} & \mbox{(by definition of $\psiT$)}\\
  &=&(\psiS{s_1} \join \psiS{s_2})\join \psiO{s} & 
     \mbox{(because $\cs{s}= \cs{s_1} \join \cs{s_2}$ by Figure~\ref{fig:subjectivity},
            and $\psi$ distributes over $\join$)}\\
  &=&\psiS{s_1} \join (\psiS{s_2} \join \psiO{s}) & \mbox{(by associativity of $\join$)}\\
  &=&\psiS{s_1} \join \psiO{s_1} & 
     \mbox{(because $\co{s_1} = \cs{s_2}\join \co{s}$ by
                                   Figure~\ref{fig:subjectivity}, 
            and $\psi$ distributes over $\join$)}\\ 
  &=&\psiT{s_1} & \mbox{(by definition of $\psiT$)}
\end{array}
\]
Notice that the proof of the property $\psiT{s} = \psiT{s_1}$
\emph{doesn't rely on the definition of $\psi$}, but only on $\psi$
being a morphism (with a trivial \seprel). Thus, the above is a
general property of morphisms that follows because the join of
\emph{self} and \emph{other} components are \emph{invariant for parent
  and children states}. In this particular
proof of $\progLock$, it allows replacing $\psiT{s_1}$ in line 6 with
$\psiT{s}$ in line 7.

The second critical point is that
$\sgS{s}=\sgS{s_1} \cupdot \sgS{s_2}$. This holds because $\sigma$ is
a morphism, and $s = s_1 \star s_2$ implies that
$\cs{s} = \cs{s_1} \join \cs{s_2}$ by Figure~\ref{fig:subjectivity},
so $\sigma$ can distribute over $\join$. Thus
$\sgS{s} = (\psiT{s} \hmapsto \currentT) \cupdot k$. Again, in
this argument we 
didn't rely on the definition of $\sigma$. 
%
%

\subsection{Morphisms as Functional Abstractions}
\label{s:morphabst}
We next proceed to transform the concrete specs of $\progLock$ and
$\progUnlock$ into specs using a morphism $\alpha : U \to O$ to more
abstractly express lock ownership. We define $\alpha$ as follows.
\begin{align}
 &\fapp{\alpha}{x} \eqdef 
      \begin{cases}
      	\own & \text{if } t \hmapsto \currentT 
		\in \fapp{\sigma}{x} \text{ for some } t\\
        \nown & \text{otherwise}
      \end{cases}      \label{eqn:morphalpha}
\end{align}
As before, in the definition of $\alpha$, one should think of $x$ as
the \emph{self} component of a thread.
Then the definition says that the thread owns the lock iff it holds
a ticket labeled $\currentT$
in the \emph{self} set of tickets ($\fapp{\sigma}{x}$).

\subsubsection*{Structure Preservation and Partiality of $\alpha$}
Just like the morphism properties of $\sigma$ and $\psi$ were
important for the internal specs to behave correctly under framing, so
any spec using $\alpha$ requires $\alpha$ to be a morphism. And
indeed, $\alpha$ satisfies the equation (\ref{eq:mp3}) from
Section~\ref{s:introseprel}.
In particular, $\fapp{\alpha}{(x \join y)}$ is defined and
$\fapp{\alpha}(x\join y) = \fapp{\alpha}{x} \join \fapp{\alpha}{y}$
but only under the condition that $x$ and $y$ don't both contain a
ticket labeled $\currentT$. In the latter case $\fapp{\alpha}{x} =
\fapp{\alpha}{y} = \own$ so their join is undefined.

A formal way to say this is that $\alpha$ is associated with the
following \seprel, where $\numCurrent{a}$ equals the number of
$\currentT$ tickets in the ticket map $a$. 
\begin{eqnarray}\label{eqn:seprelalpha}
  \label{eqn:botalpha}
  x \rel{\alpha} y \eqdef \numCurrent{\fapp{\sigma}{x}} + \numCurrent{\fapp{\sigma}{y}} \leq 1 \land x \orth y
\end{eqnarray}
The definition directly captures that together $x$ and $y$ contain at most one
served ticket. We shall see in 
Section~\ref{s:seprels} that $\rel{\alpha}$ is
indeed a \seprel, 
%
%
and moreover (Example~\ref{ex:serve}) that $\numCurrent{-}$ itself is a
morphism, composed out of map filter and map counter functions, both of which
are morphisms. 
%
%

For now it suffices to observe that if we want to use $\alpha$ in a
Hoare triple, then, at the very least, we must also attach the
property $\cs{s}\rel{\alpha}\co{s}$ to the pre- and
postcondition.\footnote{We'll see in Section~\ref{s:inversion} that
  we'll also require $\alpha$ to be an \emph{invertible} morphism, but
  that property is tied to $\alpha$ and needn't appear in Hoare
  triples.} Otherwise we won't be able to derive the framed Hoare
triples generically, i.e., by using only the property that $\alpha$ is
a morphism, without relying on $\alpha$'s definition. Framing
essentially relies on a morphism distributing over $\join$, as we've
previously seen for $\sigma$ and $\psi$, and the distribution of
$\alpha$ is conditional upon $\cs{s}\rel{\alpha}\co{s}$.


\subsubsection*{Deriving Abstract Specs.}
We thus continue to establish the following abstract, but still
intermediate, types of $\progLock$ and $\progUnlock$ via $\alpha$.
%
\begin{align*}
\HTj{\progLock}
    {&~\spec{\lambda s\ldot \fapp{\cs{\alpha}}{s} = \nown \land
       \cs{s} \rel{\alpha} \co{s}}}
    {\spec{\lambda s\ldot \fapp{\cs{\alpha}}{s} = \own \land 
       \cs{s} \rel{\alpha} \co{s}}}
    {\rTicket}\\
\HTj{\progUnlock}
    {&~\spec{\lambda s\ldot \fapp{\cs{\alpha}}{s} = \own \land
       \cs{s} \rel{\alpha} \co{s}}}
    {\spec{\lambda s\ldot \fapp{\cs{\alpha}}{s} = \nown \land
       \cs{s} \rel{\alpha} \co{s}}}
    {\rTicket}
\end{align*}
The derivations follow straightforwardly from the large footprint
specs for $\progLock$ and $\progUnlock$.
%
%
Below we just present the proof outline for $\progLock$; the
one for $\progUnlock$ is in~\acite{s:apxunlock}{\appUnlock}{POPL21arxiv}.
%
%
\newcounter{lockcodealphafour}
\def\lineno{\stepcounter{lockcodealphafour}\textsc{\thelockcodealphafour}}
\[\begin{array}[t]{r@{\quad}l}
 \lineno.& \proofspec{
 		\cs{s} \rel{\alpha} \co{s}
 		}\\ 
\lineno.& \proofspec{\sgS{s} = k 
                \land \cs{s} \rel{\alpha} \co{s} 
 		}\\
\lineno.& {\progLock}\\
\lineno.&  \proofspec{
              \sgS{s} = (\psiT{s} \hmapsto \currentT) \cupdot k 
              \land \cs{s} \rel{\alpha} \co{s}
                }\\
\lineno.& \proofspec{
             \fapp{\cs{\alpha}}{s} = \own 
             \land \cs{s} \rel{\alpha} \co{s}}
\end{array}\]
\noindent Line 1 weakens the desired precondition by eliding that
$\fapp{\cs{\alpha}}{s} = \nown$, as this property isn't actually
required by the proof. Indeed, if $\progLock$ is invoked by a thread
that already holds the lock, i.e., where
$\fapp{\cs{\alpha}}{s} = \own$, the (partial correctness) Hoare triple
for $\progLock$ holds trivially because $\progLock$ diverges.
Line 2 snapshots $\sgS{s}$ into $k$, and gives the large footprint
precondition for $\progLock$ conjoined with $\cs{s} \rel{\alpha}
\co{s}$. The latter property is an \emph{invariant} of the resource
$\rTicket$. In other words, it isn't only stable under interference of
other threads, but also it's preserved by the actions of our own
thread, as we show
in~\acite{s:stability}{\appStability}{POPL21arxiv}. 
In particular, $\cs{s}\rel{\alpha}\co{s}$ can strengthen the
precondition and weaken the postcondition of any well-typed program
that has $\rTicket$ as its resource type. In the ambient type
theory~\cite{NanevskiBDF+oopsla19} this is formally captured by a
variant of the standard Hoare logic rule of consequence that admits
the conjunction of resource invariants to the assertions in a Hoare
triple.
Therefore, line 4 follows from the large footprint spec for
$\progLock$ and the above invariance property.
%
Finally, line 5 follows immediately by the definition of $\alpha$.


\subsection{Sub-PCM}\label{s:examplesubpcm}

\subsubsection*{Construction Overview}
To obtain the ultimately desired compact specs~(\ref{spec:lock}) and
(\ref{spec:unlock}) our algebraic approach provides the sub-PCM
construction.
The construction mods out the PCM $U$ by $\rel{\alpha}$, to obtain a
sub-PCM $\UTicket$. Two ticket maps $x, y \in \UTicket$ are considered
disjoint only if $x \rel{\alpha} y$, i.e., if $x$ and $y$ have at most
one $\currentT$ ticket in total.
\begin{equation*}
 \UTicket \eqdef U / {\rel{\alpha}}
\end{equation*}
In $\UTicket$, the $\join$ operation restricts that of $U$ so that
$x \join_{\UTicket} y$ equals $x \join_U y$ if $x \rel{\alpha} y$, and
is undefined otherwise. Consequently, $x \orth_{\UTicket} y$ iff
$x \rel{\alpha} y$.
Therefore, the relation $\rel{\alpha}$ is the default notion of
\separateness in $\UTicket$. It's thus assumed of every state, and
doesn't need to be explicitly listed in any assertion.

It's essential for the sub-PCM construction that the condition by
which we mod out be a \seprel, otherwise $\join_{\UTicket}$ won't be
commutative, associative, and admit a unit. But once we know that the
condition is a \seprel, there is a generic proof
(Section~\ref{s:subpcm}) that the construction results in a PCM.
We also note that morphisms play a role in relating a PCM $U$ and a
sub-PCM $U/R$, for a \seprel $R$. As customary in algebraic
definitions of substructures, the sub-PCM construction comes with two
morphisms: \emph{injection} $\iota : U/R \to U$ and \emph{retraction}
$\rho : U \to U/R$ that allow us to transfer values and reasoning
between $U$ and $U/R$.  We explain the properties of $\iota$ and
$\rho$ in Section~\ref{s:subpcm}. In our case, the sub-PCM $\UTicket$
comes with the attendant injection $\iota_\rTicket: \UTicket \to U$,
and retraction $\rho_\rTicket: U \to \UTicket$. The injection is a
morphism with the \seprel $\orth_{\UTicket}$, and the
retraction is a morphism with \seprel $\rel{\alpha}$.

\subsubsection*{Use in Specifications}
Once we obtain the restricted PCM $\UTicket$, we can proceed to
construct a \emph{sub-resource} $\rTicket'$ which restricts the PCM
$U$ of $\rTicket$ to $\UTicket$. The formal discussion of resources is
given in~\cite{NanevskiBDF+oopsla19}. Here, we just mention that
$\rTicket'$ \emph{simulates} $\rTicket$, intuitively, because each
transition of $\rTicket$ preserves $\rel{\alpha}$. The latter is easy
to check: if in a state $s$ the map $\sgT{s}$ has at most one
$\servedT$ ticket, then so does a state $s'$ obtained by executing one
of the transitions of $\rTicket$ in $s$.

The ambient type theory provides an inference rule by which one can
compositionally change the resource of a program from $\rTicket$ to
$\rTicket',$\footnote{Or to any resource simulating $\rTicket$.} while
precomposing the morphisms in the specs with the injection
$\iota_\rTicket$.
Thus, we can transform the previous specs using $\alpha$ into the ones
given below where $\alpha' = \alpha \circ \iota_\rTicket$. The
condition $\cs{s} \rel{\alpha} \co{s}$ transforms into
$\cs{s} \orth_{\UTicket} \co{s}$ and can thus be
elided.
%
%
This yields the specs we set out to obtain, modulo the renaming of
$\alpha$ and $\rTicket$ into $\alpha'$ and $\rTicket'$.
\begin{align*}
\HTj{\progLock}
    {\spec{\lambda s\ldot \papp{\cs{\alpha'}}{s} = \nown}
	}
    {\spec{\lambda s\ldot \papp{\cs{\alpha'}}{s} = \own}
	}
    {\rTicket'}\\
\HTj{\progUnlock}
    {\spec{\lambda s\ldot \papp{\cs{\alpha'}}{s} = \own}
	}
    {\spec{\lambda s\ldot \papp{\cs{\alpha'}}{s} = \nown}
	}
    {\rTicket'}
\end{align*}

We emphasize that the simple Hoare specs are not the only benefit of
the sub-PCM construction. 
By constructing $\UTicket$, we not only restricted the states of
$\rTicket$, but we did so in a way that promoted $\rel{\alpha}$ into
the new default notion of \separateness. 
%
Thus, we can now reason about $\rel{\alpha}$ using the support that
the ambient type theory provides for \separateness in the form of
lemma libraries and decision procedures, and which wouldn't have
applied if $\rel{\alpha}$ is simply listed as a conjunct in the
assertions, and tracked as just another hypothesis in the proof
state. We shall see in Section~\ref{s:morph} that the
$\mathsf{ordered}$ property can also be viewed as a \seprel and thus
moved from the state space $\Sigma_{\rTicket}$ into the PCM by a
sub-PCM construction. On the other hand, $\mathsf{no\_gaps}$ doesn't
admit such a move. We demonstrate in Section~\ref{s:morph} that
$\mathsf{no\_gaps}$ isn't a \seprel; it doesn't generalize a disjointness
relation between states of two threads, but rather represents a global
property of the aggregated state of all threads.

It's also worth mentioning that we could have obtained the above specs
in several alternative ways. 
%
For example, we could have started our example immediately by using
$\UTicket$ instead of $U$. Correspondingly, instead of $\sigma$ and
$\psi$, we would have used $\sigma' = \sigma \circ \iota_\rTicket$ and
$\psi' = \psi \circ \iota_\rTicket$ in our specs and proof outlines. The whole
development that we carried out in this section then retraces
easily. This shows that the approach is flexible enough to achieve the
same specs and proofs by different order and arrangement.

We could also have chosen a different internal representation
altogether. For example, we could have stored the tickets not into a
map, but into three disjoint sets: one set for $\waitingT$, one for
$\currentT$, and one for $\servedT$ tickets, with the restriction that
the set for $\currentT$ tickets has at most one element. The algorithm
would then shuffle tickets between sets to track the progress of the
ticket through the bakery.
A PCM implementing this alternative representation would be
\emph{isomorphic} to $\UTicket$. But, to be able to formally speak of
PCM isomorphism, one first has to have a notion of PCM
\emph{morphism}, as it applies to separation logic. Developing such
a notion, along with the associated notion of separating relation, is the contribution of this paper.



\section{PCM Abstractions Formally}
\label{s:morph}


\subsection{Making Partiality Explicit}


In previous sections, our discussion of partiality has been implicit,
as we merely postulated that $\join$ and various PCM morphisms are
allowed to be undefined on some inputs.
In practical formalization, however, it's useful to make partiality
explicit by enriching the carriers with a new element that a function
returns whenever it's supposed to be undefined. This is a common
practice in theories of partial functions, e.g., domain
theory~\cite{abramsky+jung}, and in symbolic execution in separation
logic~\cite{BerdineCO05}.
We make a similar enrichment here as well.

\begin{definition}
\label{def:TPCM}
A \dt{topped partial commutative monoid} is a 5-tuple $(\pcmA,
\join, \unit, \undefOp, D)$ such that $\join$ is a \emph{total}
commutative and associative operation on $A$, with $\unit$ as the
unit. The element $\undefOp \in A$ is the canonical \emph{undefined}
element, and $D \subseteq \pcmA$ is the set of \emph{defined}
elements. The two satisfy the following properties.
\begin{enumerate}
\item $\undefOp \notin D$
\item $\unit \in D$ 
\item if $x \join y \in D$ then $x, y \in D$
\item $x \join \undefOp = \undefOp \join x = \undefOp$ 
\end{enumerate}
We say that a topped PCM is \dt{normal}, if $A = D \cup
\{\undefOp\}$, i.e., $\undefOp$ is the only undefined element.
\end{definition}

The definition introduces the element
$\top$ which functions are supposed to return to signal undefinedness.
For technical reasons that we explain below, we allow topped PCMs that
have multiple undefined elements, but
$\undefOp$ is a distinct one among them, and in particular, (1)
$\undefOp \notin D$.
The above properties further say that (2) $\unit$ is defined, and that
(3) a join with an undefined element must be undefined. More strongly,
(4) a join with $\undefOp$ must equal $\undefOp$, i.e., $\undefOp$ is
the \emph{absorbing} element of $A$ (also known as the \emph{zero}).
We continue to say that $x$ and $y$ are \separate, written
$x \orth y$, if $x \join y$ is defined, i.e., $x \join y \in D$. It's
easy to see that $x \in A$ is defined iff it's \separate from
$\unit$. Thus we write $x \rel{} \unit$ instead of $x \in D$ to say
that $x$ is defined.
As all the PCMs we consider in this paper are topped, we
dispense with the adjective.

%

\begin{example}
\label{ex:finmap}
	The PCM $\mapTctl$, which we used to represent the internal state of ticket
	locks is an instance of a more general PCM of finite maps. 
	Finite maps $A \finmap B$ form a topped normal PCM:
	take $(A \finmap B) \cup \set{\undefOp}$ as the carrier,
        $D = A \finmap B$ for the defined elements, the empty map
        $\emptyset$ as the unit, and the join defined as
        \begin{align*}
   	f \join g \eqdef 
		\begin{cases}
                f \cup g & \text{if $f, g \neq \undefOp$ and $f$, $g$
                  are maps with disjoint domains}\\
		\undefOp & \text{otherwise}
		\end{cases}
	\end{align*}
\end{example}

%

\begin{proposition}
\label{prop:prodPCM}
 Given (topped) PCMs $\pcmA$, $\pcmB$, the \dt{Cartesian product} $A
 \times B$ is a (topped) PCM with $\join$, $\unit$ and $\undefOp$
 defined pointwise: $(a_1, b_1) \join{} (a_2, b_2) \eqdef (a_1
 \join_{\pcmA} a_2, b_1 \join_{\pcmB} b_2)$, $\unit \eqdef
 (\unit_\pcmA, \unit_\pcmB)$ and $\undefOp \eqdef (\undefOp_\pcmA,
 \undefOp_\pcmB)$, and the set of defined elements $D = D_\pcmA \times
 D_\pcmB$.
\end{proposition}

The above proposition shows that $V = \pcmA \times \pcmB$ is a topped
PCM whenever $\pcmA$ and $\pcmB$ are, but $V$ isn't necessarily
normal.
%
%
%
Indeed, $V$ contains elements of the form $(a, \undefOp_\pcmB)$ and
$(\undefOp_\pcmA, b)$, where $a \in D_\pcmA$ and $b \in
D_\pcmB$. These elements can't be defined (hence, neither is in
$D_{V}$), but they're all distinct from
$\undefOp_{V} = (\undefOp_\pcmA, \undefOp_\pcmB)$.
The common way to avoid the proliferation of undefined elements in
theories of partiality is to consider \emph{smash products} instead of
Cartesian products. In this paper, we instead allow PCMs such as
Cartesian products that aren't normal. We also allow morphisms on
them, e.g., the projections $\pi_1 : V \rightarrow \pcmA$ and
$\pi_2: V \rightarrow \pcmB$.
When required, we rely on the sub-PCM construction (to be defined
shortly) to normalize a PCM by removing the undefined elements other
than $\undefOp$.

\subsection{\SepRels} 
\label{s:seprels}
We next define \emph{\seprels}, a strengthening of disjointness
$\rel{}$ of the underlying PCM. 
Having in
mind that our specs apply \seprels to \emph{self} and \emph{other}
components of a state, one can thus view \seprels as determining when
two PCM values can be used to model the state of two concurrent
threads.


\begin{definition}\label{def:seprel}
  Relation $\genrel{}{}$ on the carrier of the PCM $\pcmA$, is a
  \dt{\seprel} if it satisfies the following laws which make $R$ a structure-preserving relation on $A$.
\begin{enumerate}
 \item \emph{(definedness)} if $\genrel{x}{y}$ then
   $\genrel{x}{\unit}$
 \item \label{seprel:strngth}\emph{(strengthening)} if $\genrel{x}{y}$ then $x \orth y$
 \item \emph{(unit)} $\genrel{\unit}{\unit}$
 \item \emph{(symmetry)} $\genrel{x}{y}$ iff $\genrel{y}{x}$
 \item \emph{(associativity)} if $\genrel{x}{y}$ and $\genrel{(x \join
   y)}{z}$ then $\genrel{x}{(y \join z)}$ and $\genrel{y}{z}$
\end{enumerate}
\end{definition}

The law (1) restricts the \seprel $R$ to defined elements only, as
only a defined element should represent the state of a thread. Law (2)
says that $R$ strengthens the \separating relation of the underlying
PCM. Law (3) says that empty state is a valid state for any two
threads, and law (4) says that the order in which threads appear in
the relation is irrelevant.

The associativity law (5) describes when we can transfer ownership of
state between two threads. Let's assume that we have two concurrent
threads $\theta_1$ and $\theta_2$. Correspondingly, their states are
related by $R$. Let $z$ be the state of $\theta_2$, and let $\theta_1$
be a parent of two other concurrent threads with states $x$ and $y$,
respectively. Thus $x~R~y$ and $(x \join y)~R~z$. The law says that we
can transfer $y$ from $\theta_1$ to $\theta_2$, which essentially
corresponds to re-associating the child of $\theta_1$ owning $y$ to
$\theta_2$. Intuitively, this is possible because the ordering and
grouping of the threads in a thread pool is irrelevant.

%

Notice that from $x~R~y$ and $(x \join y)~R~z$, by symmetry of $R$ and
commutativity of $\join$, we get $y~R~x$ and $(y \join x)~R~z$, which by
associativity implies $y~R~(x \join z)$ and $x~R~z$ as well. Thus,
it's convenient to introduce the following notation for the antecedent
of the associativity law:
\[
x ~R~ y ~R~ z \eqdef x ~R~ y \land (x \join y)~R~z
\]
%
to say that $x$, $y$ and $z$ represent states of three concurrent
threads, which are pairwise \separate, and each is separate from the
join of the other two.  

%


\begin{proposition}
\label{p:sepU0}
Let $U$ be a PCM, and $R$ a \seprel on $U$. Then
$\fapp{x}{\fapp{R}{y}}$ implies $\fapp{(x \join y)}{\fapp{R}{\unit}}$.
\end{proposition}
\begin{proof}
  From $\fapp{x}{\fapp{R}{y}}$ we derive $\fapp{\unit}{\fapp{R}{x}}$
  by the definedness and symmetry laws for $R$, and
  $\fapp{(\unit \join x)}{\fapp{R}{y}}$ because $\unit$ is the
  unit. Then by associativity $\fapp{\unit}{\fapp{R}{(x \join y)}}$,
  and by symmetry $\fapp{(x \join y)}{\fapp{R}{\unit}}$.
%
\end{proof}
%


%
%
The proof of Proposition~\ref{p:sepU0} uses associativity, and we can
explain the proposition using threads similarly to how we explained
associativity. The proposition says:
if $x, y$ are valid
states of two concurrent threads, 
then joining them produces a parent whose state $(x \join y)$ is valid
as well.



\subsubsection*{Basic Examples of \SepRels}
%
The smallest \seprel of a PCM $A$ is induced by $A$'s unit. We denote
it $\rel{\unit_A}$, and define it by
\[x \rel{\unit_A} y \eqdef x = \unit_A \land y = \unit_A\]
The relation clearly satisfies the required laws. 
Similarly, the PCM $A$ itself induces the \dt{trivial} \seprel
$\rel{}$ (or $\rel{A}$ when we want to make $A$ explicit), defined as
\[
x \rel{A} y \eqdef (x \join y)\ \mbox{is defined}
\]
This is the largest \separating relation on $A$, since any larger
relation violates the strengthening property (2).
The intersection of two \separating relations is also a separating
relation.
The \emph{join} relation $J$ on $\pcmA^2$ defined as $(a_1,
a_2)\,J\,(b_1, b_2)$ iff $(a_1 \join a_2) \rel{} (b_1 \join b_2)$ iff
$(a_1 \join a_2 \join b_1 \join b_2) \rel{} {\unit_A}$ is a separating
relation.

\subsubsection*{\SepRel $\rel{\alpha}$}
%
Our formalization of ticket locks in Section~\ref{s:example} uses the
invariant $x \rel{\alpha} y$ on the PCM $U$, defined
in~(\ref{eqn:botalpha}) to state that at most one of the maps
$\fapp{\sigma}{x}$ and $\fapp{\sigma}{y}$ holds the $\currentT$
ticket. The property ensures that $\alpha$ is a morphism.

It's easy to see that $\rel{\alpha}$ is a \seprel. The conditions
(1-4) of Definition~\ref{def:seprel} are immediate.
To show associativity, assume that $(x \join y) \rel{\alpha} z$ (we
don't need $x \rel{\alpha} y$). Because $\numCurrentX$ distributes
over $\join$ (to be shown in a more general form in
Example~\ref{ex:serve}), the assumption gives us
\[
(\numCurrent{\fapp{\sigma}{x}} + \numCurrent{\fapp{\sigma}{y}}) +
\numCurrent{\fapp{\sigma}{z}} \leq 1
\]
But then
$\numCurrent{\fapp{\sigma}{x}} + (\numCurrent{\fapp{\sigma}{y}} +
\numCurrent{\fapp{\sigma}{z}}) \leq 1$ and
$\numCurrent{\fapp{\sigma}{y}} + \numCurrent{\fapp{\sigma}{z}} \leq
1$, i.e., $x \rel{\alpha} (y \join z)$ and $y \rel{\alpha} z$.
We see that associativity in this example says that three threads
may group in any way while preserving $\rel{\alpha}$ because at most
one of them can hold the $\currentT$ ticket.  

%
%

\subsubsection*{\SepRel $\ord$}
In Section~\ref{s:transitions}, we defined the state space $\STicket$
of the resource \rTicket using the predicate $\ord$ to capture that
$\servedT$ tickets are smaller than $\currentT$ ticket, which in turn
is smaller than $\waitingT$ tickets. While $\ord$ is defined as a
predicate over a single PCM element $x \in U$, it easily lifts to a
binary relation as follows:
\[
x\ \omega\ y \eqdef \fapp{\ord}{(x \join y)} \land x \orth y 
\]
It's easy to see that $\omega$ is a \seprel; again, the key property
is associativity: $(x\ \omega\ y)$ and $(x \join y)\ \omega\ z$ imply
$(y\ \omega\ z)$ and $x\ \omega\ (y \join z)$. By definition of
$\omega$, we must show: $\fapp{\ord}{(x \join y)}$ and $\fapp{\ord}{(x
  \join y \join z)}$ together imply $\fapp{\ord}{(y \join z)}$ and
$\fapp{\ord}{(x \join y \join z)}$.  
This holds because if a map is $\ord$, then trivially, every subset of
that map is $\ord$ as well.
The conjunct $x \orth y$ ensures the strengthening property
(\ref{seprel:strngth}) of Definition~\ref{def:seprel}.  Thus, our
construction of $\rTicket'$ could have moved $\ord$ from the state
space $\STicket$ into the definition of the sub-PCM $\UTicket$.


\subsubsection*{Non-Example of \SepRel: $\nogaps$}
The state space $\STicket$ also used the predicate $\nogaps$ to
capture that the tickets are drawn consecutively starting from
ticket~$1$. Similarly to $\ord$, $\nogaps$ can be made into a binary
relation:
\[
x\ \upsilon \ y \eqdef \fapp{\nogaps}{(x \join y)} \land x \orth y 
\]
In contrast to $\omega$, however, the relation $\upsilon$ isn't
associative, and hence isn't a \seprel. For example, let $x$, $y$ and
$z$ be ticket maps with domains $\dom{x} = \{2\}$, $\dom{y} = \{1\}$,
and $\dom{z} = \{3\}$, respectively.
%
%
Then $x\ \upsilon\ y$ and $(x \join y)\ \upsilon\ z$ hold because
$\dom{x \join y} = \{1, 2\}$ and
$\dom{x \join y \join z} = \{1,2,3\}$ contain consecutive tickets
starting from ticket $1$. But clearly $y\,{\cancel\upsilon}\,z$
because $\dom{y \join z} = \{1, 3\}$ has a gap, missing the ticket $2$ (see
Figure~\ref{fig:nogaps}). 

In this sense, $\nogaps$ is a \emph{global} property. It holds of the
collective state of all threads taken together, but, unlike \seprels,
doesn't relate any two individual threads. In particular, $\nogaps$
can't be moved from $\STicket$ into $\UTicket$. In other words, PCMs
and \seprels encode local properties of thread states, while resource
state spaces encode global ones.
%
%

\begin{figure}
\[
\begin{array}{c@{\quad\quad}c@{\quad\quad}c}
\begin{array}{c}
\begin{tikzpicture}[scale=0.94, transform shape]
  \filldraw[fill=gray!70, draw=black] (0, 0) rectangle (-0.6, -0.7);
  \filldraw[fill=gray!50, draw=black] (0, -0.7) rectangle (-0.6, -1.4);
  \filldraw[fill=gray!20, draw=black] (0, -1.4) rectangle (-0.6, -2.1);
  \node[scale=0.9] at (-0.3, -0.35) {$z$};
  \node[scale=0.9] at (-0.3, -1.05) {$x$};
  \node[scale=0.9] at (-0.3, -1.75) {$y$};
\end{tikzpicture}\\
{\fontsize{9}{9}\textrm{$x \join y \join z$}}
\end{array}
&
\begin{array}{c}
\begin{tikzpicture}[scale=0.94, transform shape]
  \filldraw[fill=white,draw=none] (0, 0) rectangle (-0.6, -0.7);
  \filldraw[fill=gray!50, draw=black] (0, -0.7) rectangle (-0.6, -1.4);
  \filldraw[fill=gray!20, draw=black] (0, -1.4) rectangle (-0.6, -2.1);
  \node[scale=0.9] at (-0.3, -1.05) {$x$};
  \node[scale=0.9] at (-0.3, -1.75) {$y$};
\end{tikzpicture}\\
{\fontsize{9}{9}\textrm{$x \join y$}}
\end{array}
&
\begin{array}{c}
\begin{tikzpicture}[scale=0.94, transform shape]
  \filldraw[fill=gray!70, draw=black] (0, 0) rectangle (-0.6, -0.7);
  \filldraw[fill=gray!20, draw=black] (0, -1.4) rectangle (-0.6, -2.1);
  \node[scale=0.9] at (-0.3, -0.35) {$z$};
  \node[scale=0.9] at (-0.3, -1.75) {$y$};
\end{tikzpicture}\\
{\fontsize{9}{9}\textrm{$y \join z$}}
\end{array}
\end{array}\vspace{-1em}
\]
\caption{Binary relation over predicate $\nogaps$. Adjacent elements satisfy the
relation.}
  \label{fig:nogaps}\vspace{-3mm}
\end{figure}
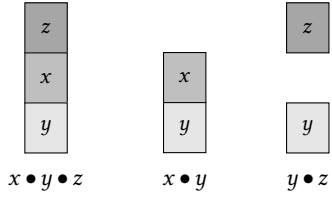

\subsection{Morphisms}
\label{sec:morphisms}
\begin{definition}\label{def:pcmmorph}
  A \dt{morphism} $\phi : \pcmA \to \pcmB$
  with a \separating relation $\rel{\phi}$ is a function $\phi$ from
  $\pcmA$ to $\pcmB$ such that.
\begin{enumerate}
\item \emph{(preservation of unit)} $\fapp{\phi}{\unit_\pcmA} = \unit_\pcmB$
\item \emph{(preservation of undefinedness)} $\fapp{\phi}{\undefOp_\pcmA} = \undefOp_\pcmB$
\item \emph{(distributivity)} if $x \rel{\phi} y$ then 
			${\fapp{\phi}{x}}\orth{\fapp{\phi}{y}}$, and
			$\papp{\phi}{x \join y} = \fapp{\phi}{x} \join \fapp{\phi}{y}$
\end{enumerate}
We say that $\phi$ is a \dt{total} PCM morphism if $\rel{\phi}$ equals
the trivial \seprel $\rel{\pcmA}$.
\end{definition}
%


Some basic examples of morphisms include the identity function
$\id_\pcmA : \pcmA \to \pcmA$, which is a total morphism on $\pcmA$.
%
%
So is the always-unit function $- \mapsto \unit_B : \pcmA \to \pcmB$,
as are the projections out of the product PCM. 
%
%
We also have the morphism $\mathsf{join}_A : \pcmA^2 \to \pcmA$
defined as $\fapp{\mathsf{join}_A}{(a, b)} = a \join b$, which is a
morphism under the \separating relation $J$ (Section~\ref{s:seprels}).


Morphisms and \seprels support a number of common algebraic
constructions.

\begin{definition} 
  \label{def:constrs}
  Let $\alpha$, $\beta$ be PCM morphisms. The
  \dt{composition} $\comp{\alpha}{\beta}$, \dt{tensor product}
  $\tepr{\alpha}{\beta}$, and 
  \dt{arrow product} $\alpha \times\beta$ are
  defined as below. All three are morphisms, under the given
  \separating relations.
\begin{align*}
\fapp{(\comp{\alpha}{\beta})}{x} & \eqdef \papp{\alpha}{\fapp{\beta}{x}}
 & \text{with} &&
x \rel{\comp{\alpha}{\beta}} y & \eqdef x
  \rel{\beta} y \land \fapp{\beta}{x} \rel{\alpha} \fapp{\beta}{y}
\\
\fapp{(\tepr{\alpha}{\beta})}{x} & \eqdef (\fapp{\alpha}{x}, \fapp{\beta}{x})
& \text{with} &&
x \rel{\tepr{\alpha}{\beta}} y & \eqdef 
x \rel{\alpha} y \land x \rel{\beta} y \\
\fapp{(\alpha \times \beta)}{(x_1, x_2)} & \eqdef (\fapp{\alpha}{x_1},
                                     \fapp{\beta}{x_2}) 
& \text{with} && (x_1, x_2)\,\orth_{\alpha\times\beta}\,(y_1, y_2) & \eqdef
x_1\,\orth_\alpha\,y_1 \wedge x_2\,\orth_\beta\,y_2
	\end{align*}
\end{definition}



We can also define kernels, equalizers and restrictions of PCM
morphisms, as customary in various algebraic theories. We don't apply
these constructions in the ticket lock example but comment below why
they are useful.
Importantly, our theory is closed under these constructions, as
equalizers and kernels of morphisms are \seprels, and a restriction of
a morphism by a \seprel is a morphism. This shows that \seprels and
morphisms are natural notions to consider together. Moreover:

\begin{theorem}
  \label{thm:cat}
  Morphism composition is associative, with the identity morphism as
  unit. Thus, the structure with (topped) PCMs as objects and PCM
  morphisms as arrows, forms a category.
\end{theorem}


\begin{definition}
  \label{def:kereq}
  Let $\alpha$ and $\beta$ be PCM morphisms. The kernel
  $\ker \alpha$ and equalizer $\eql\,\alpha\,\beta$ are defined
  below. Both are \seprels.
	\begin{itemize}
		\item $x\,(\fapp{\ker}{\alpha})\,y \eqdef \relAlpha x y \land \fapp{\alpha}{x} =
			\fapp{\alpha}{y} = \unit$, and
		\item $x\,(\eql\,\alpha\,\beta)\,y \eqdef x \rel{\alpha} y \land x
			\rel{\beta} y \land \fapp{\alpha}{x} = \fapp{\beta}{x} \land \fapp{\alpha}{y} =
			\fapp{\beta}{y}$.
	\end{itemize}
\end{definition}

Equalizers are useful whenever one wants to equate components of a
PCM. For example, it's common in practice to have PCMs $A$ and $B$,
and to consider pairs $(a, b) \in A \times B$, but only under the
restriction that $\phi(a) = \psi(b)$ for some morphisms
$\phi : A \to X$ and $\psi : B \to X$. The morphisms $\phi$ and $\psi$
would typically be projections, thus forcing that $A$ and $B$ are
``stitched'' along the projected components.
The desired pairs are described by the equalizer
$\fapp{\fapp{\eql}{(\phi\times
    \iota_B)}}{(\iota_A\times\psi)}$. Kernels are a special case of
equalizers, when one of the morphisms is the always-unit function. 





\begin{definition}
 \label{d:restriction}
  A \dt{restriction} of a morphism $\alpha$ with \seprel $R$ is
  defined below. It's a morphism under the given \seprel.
	\begin{gather*}
		\fapp{(\alpha / R)}{x} \eqdef 
		\begin{cases}
			\fapp{\alpha}{x} & \text{if } x ~R~ \unit \\
			\undefOp & \text{otherwise}
		\end{cases} 
		\text{\quad with}\quad x\,\rel{\alpha / R}\,y \eqdef x \rel{\alpha} y \land x ~R~ y
	\end{gather*} 
\end{definition}


Returning to our leading example of ticket lock, we can identify several other
examples of morphisms that we used.

Section~\ref{s:morphabst} mentioned that $\numCurrent{-}$ is a morphism. This morphism, however, can be decomposed into two simpler morphisms.

\begin{example}
  \label{ex:serve}
  The function $\texttt{filter}_l : (A \finmap B) \to (A \finmap B)$
  over a finite map selects only the entries that map to the label
  $l$.
$$
	\fapp{\texttt{filter}_l}{m} \eqdef
		\set{z \hmapsto l \mid \fapp{m}{z} = l}
		\quad \text{with } m \rel{{\texttt{filter}}_l} n \eqdef \valid m n
$$
\noindent The $\texttt{filter}$ function is a total morphism.
Similarly, the counting function $\# : (A \finmap B) \to \Nat$
computing the size of the domain of a finite map is a total
morphism. Then, we can define
$\numCurrent{x} = \papp{(\comp{\#}{\texttt{filter}_\currentT)}}{x} $
and since morphisms compose, it's a morphism as well.
\end{example}

\subsection{Sub-PCMs}
\label{s:subpcm}

In Section~\ref{s:examplesubpcm}, we restricted the PCM $U$ with a
\seprel $\relAlpha {} {}$. Formally, this construction is developed as
a sub-object of a PCM, a \emph{sub-PCM}. As customary in abstract
algebra, we present the construction through two \emph{morphisms} on
the objects; \emph{injection} of a sub-PCM into a PCM and a
\emph{retraction} from a PCM into its sub-PCM.  



\begin{definition}\label{d:subpcm}
	\label{def:subpcm}
	A PCM $\pcmA$ is a \dt{sub-PCM} of a PCM $\pcmB$ if there exists a total PCM
	morphism $\iota : \pcmA \to \pcmB$ (an injection) and a morphism $\rho :
	\pcmB \to \pcmA$ (a retraction), such that:
	\begin{enumerate}
		\item $\fapp{\rho}{(\fapp{\iota}{x})} = x$,
		\item if $u \rel{\rho} \unit$ then $\fapp{\iota}{(\fapp{\rho}{u})} = u$
		\item if $\fapp{\rho}{u} \rel{\pcmA} \fapp{\rho}{v}$ then $u \rel{\rho} v$
		\item if $\fapp{\iota}{x} \rel{\pcmB} \unit$ then $x \rel{\pcmA} \unit$
	\end{enumerate}
\end{definition}

Property (1) says that $\iota$ is injective, i.e., if we coerce $x \in
A$ into $\fapp{\iota}{x}$, we can recover $x$ back by $\rho$, since no
other element of $A$ maps by $\iota$ into $\fapp{\iota}{x}$. This is a
common property in sub-object constructions, and essentially says that
$\iota$ embeds $A$ into a subset of $B$.
The dual property (2) allows the same for the elements of $B$ that are
defined according to $\rel{\rho}$. Hence, $A$ is in 1-1 correspondence
with that subset of $B$.
Property (3) extends the correspondence between $A$ and the subset of
$B$ to their respective notions of disjointness. That is, $\orth_A$,
when considered on images under $\rho$, implies (and hence, by
properties of morphisms \emph{equals}) $\orth_\rho$.
Property (4) says that $\iota$ preserves the undefined elements, so
that embedding $A$ into $B$ doesn't accidentally turn an undefined
element into a defined one. A similar property of $\rho$ is a
consequence of (3). Finally, $\iota$ is total in order to embed the
whole of $A$ into $B$. A partial $\iota$ would embed only a subset of
$A$ into $B$, but that can be modeled by considering a total morphism
from that subset into $B$.

As a simple example, we note that $\pcmA$ is a  sub-PCM of
itself with the identity injection and retraction and trivial
separating relations.


Definition~\ref{def:subpcm} says what it means to be a sub-PCM
abstractly, in terms of morphisms and \seprels. We next proceed to give
a concrete construction that mods out a PCM $U$ by a \separateness
relation $R$ to obtain a PCM $U/R$, that is a sub-PCM of $U$ according
to Definition~\ref{d:subpcm}.  It is this construction that we used in
Section~\ref{s:example} to obtain the PCM $\UTicket$ out of $U$.
%
%
%
The construction starts by defining the carrier set $U/R$, and the
unit and $\join$ as follows.
%
\[
\begin{array}[r]{c@{\qquad}c@{\qquad}c}
U / R \eqdef \{z : U \mid \fapp{z}{\fapp{R}{\unit_U}} \} \cup \{ \top
  \} &
\unit_{U / R} \eqdef \unit_U & 
x \join_{U / R} y \eqdef 
\begin{cases}
  x \join_{U} y & \mbox{if } \fapp{x}{\fapp{R}{y}}\\
  \top & \mbox{otherwise}
\end{cases}
\end{array}
\]
%
%
\noindent Elements of the sub-PCM are the elements of $U$ that are
defined wrt.~$R$, and the unit and $\join$ are obtained by lifting the
operations of $U$. 
Notice that the operations are well-defined. In particular, $\unit_{U/R}$
is in the carrier set $U/R$, since $1_U\,R\,1_U$ by the properties of
\seprels. Also, if $x, y \in U/R$, then $x \join_{U/R} y \in
U/R$. This is proved by case analysis on whether $x$ and $y$ are
defined or $\undefOp$. The interesting case is when they're defined
and $x\,R\,y$. Then by Proposition~\ref{p:sepU0},
$(x \join y)\,R\,\unit$, so the conclusion follows immediately.


\begin{lemma} 
	\label{lem:subpcmprop}
	The definitions of $\unit_{U / R}$ and $\join_{U / R}$ satisfy the
	following properties:
	\begin{enumerate}
		\item $\join_{U / R}$ is commutative,
			i.e., $x \join_{U / R} y = y \join_{U / R} x$,
		\item $\join_{U / R}$ is associative,
			i.e., $(x \join_{U / R} y)
				\join_{U / R} z = x \join_{U / R} (y
                                \join_{U / R} z)$,
                \item $1_{U/R}$ is the unit for $\join_{U/R}$, i.e.,
                  $1_{U/R} \join_{U/R} x = x \join_{U/R} 1_{U/R} = x$,
		\item $\top \not\in \{z : U \mid z\,R\,\unit\}$,
		\item $\unit \in \{ z : U \mid z ~R~ \unit \}$, and
		\item if $x \join_{U / R} y \in \{z : U \mid z ~R~ \unit\}$ then
				$x, y \in \{z : U \mid z ~R~ \unit\}$.
	\end{enumerate}
\end{lemma}

\begin{proof}
  (Sketch.) By easy analysis of the cases in the definition of $U/R$ and
  $\join_{U/R}$. The proof essentially requires all the \seprel
  properties of $R$. For example, the commutativity property (1)
  relies on the symmetry of $R$, the associativity property (2) relies
  on the associativity of $R$, and the unit property (3) relies on the
  unit law of \seprels. Thus, \seprel laws are directly obtained as a
  requirement for proving this lemma.
\end{proof}

\noindent Thus, by the above lemma, we have a PCM:
\begin{equation*}
		(U / R , \join_{U / R}, \unit_{U / R}, \top, \{z : U
		\mid z ~R~ \unit\})
\end{equation*}
It remains to show that this PCM is a sub-PCM in the sense of
Definition~\ref{def:subpcm}. To that purpose, we define the two necessary
morphisms:
\[
\begin{array}[t]{c@{\qquad}c}
\fapp{\iota}{x} \eqdef x \mbox{ with } \rel{\iota} \eqdef \rel{U/R} &
\fapp{\rho}{u}  \eqdef 
  \begin{cases}
	u	& \mbox{if } u ~R~ \unit\\
       \undefOp & \mbox{otherwise}
  \end{cases}   \mbox{ with } \rel{\rho} \eqdef R 
\end{array}
\]
%
\noindent Note that these functions are indeed morphisms. That $\iota$
and $\rho$ preserve unit and $\undefOp$ is trivial to show, and so is
that $\papp{\iota}{x\join y} = \fapp{\iota}{x} \join \fapp{\iota}{y}$
when $x \rel{\iota} y$. It remains to show that
$\papp{\rho}{u \join v}$ is defined and
$\papp{\rho}{u \join v} = \fapp{\rho}{u} \join \fapp{\rho}{v}$, if
$u\,R\,v$. To see this, assume $u\,R\,v$ and observe that from the law
of defined elements of \seprel, this implies $u\,R\,\unit$ and
$v\,R\,\unit$. Thus $\fapp{\rho}{u}=u$ and $\fapp{\rho}{v}=v$ and so
$\fapp{\rho}{u}\join\fapp{\rho}{v} = u \join v$. By
Proposition~\ref{p:sepU0}, we also have $(u \join v)\,R\,\unit$; thus
$(u \join v)$ is defined and $\fapp{\rho}{(u \join v)}$ equals
$u \join v$, concluding that $\rho$ is a morphism. Now it's also easy
to see that the injection $\iota$ is total (by definition, since it
has the trivial \seprel), and that $\iota$ and $\rho$ satisfy the
requirements of Definition~\ref{d:subpcm}. Therefore $U/R$ is a
sub-PCM of $U$.

We conclude this section by noticing that $U/R$ is a normal PCM, since
$\undefOp$ is its only undefined element. Thus, we can use the sub-PCM
construction to normalize PCMs, when desirable. Given a non-normal PCM
$A$, the PCM $A/{\orth_A}$ is normal and contains all the defined
elements of $A$.

\subsection{Histories, Morphisms, and Separating Relations}
\label{sec:spinex}
This section illustrates how PCM morphisms and \seprels apply to
reasoning about data structures specified via time-stamped histories.
Histories are a common and general abstraction in concurrency, used,
for example, in the formulation of consistency criteria such as
linearizability~\cite{Herlihy-Wing:TOPLAS90}.
%
Here, we specifically focus on their application to locking.

An abstract locking history of a thread is a finite map from
timestamps represented by positive natural numbers to set $\Op =
\set{\lockval,\unlockval}$, i.e. $\mathsf{Hist} = \Nat^+ \finmap
\Op$. If a thread's history has the value
 $\lockval$ at timestamp $t$,
that signifies that the thread has \emph{locked} at time $t$. Similarly, if
the value is $\unlockval$ then the thread has \emph{unlocked} at time
$t$. If the history of a thread is undefined at $t$, then the thread
was idle at that moment, and some other thread may have locked or
unlocked at time $t$.
We overload the notation from Section~\ref{s:transitions}
and write $t \hmapsto \lockval$ (resp. $t
\hmapsto \unlockval$) for a singleton history containing only the
timestamp $t$ with the locking (resp. unlocking) operation.
Histories form a PCM under disjoint union,
with the nowhere defined map (i.e., empty history) as unit.


\subsubsection*{Using PCM $\Hist$ In a Resource}
%
Let us assume that we have defined an internal state of some
(unspecified) locking algorithm and a corresponding resource with a
state space and transitions, similar to how we defined the resource
$\rTicket$ in Section~\ref{s:example}.
Then we will typically have a morphism, which we name $\tau$ here,
that projects the history component of the underlying state; that is
$\tauS{s}$ is the history of ``our'' thread in state $s$, and
$\tauO{s}$ is the history of all ``other'' threads combined.

Moreover, we will also typically use the resource state space to
specify global properties of histories, similar to $\nogaps$ from
Figure~\ref{fig:ststicket}. For example, at the very least, we will
require that the global history $\tauT{s}$ alternates the operations
$\lockval$ and $\unlockval$, i.e., $\tauT{s}$ has the form
$ \tauT{s} = 1 \hmapsto \lockval \join 2 \hmapsto \unlockval \join 3
\hmapsto \lockval \join \dots$
The alternation property specifies the key relationship between
locking and unlocking, but doesn't form a \seprel itself. However, as
we show presently, there's an important \seprel $\relOmega {} {}$ over
locking histories that gives rise to a morphism $\omega :
\mathsf{Hist} \to O$ for computing lock ownership out of a thread's
history.

\subsubsection*{Histories of Exclusive Locking}
Note that histories, even with the alternation property imposed, don't
exclude the possibility that one thread may take the lock, which is
then released by another thread. In our subjective setting, we may
represent the situation as follows
\[
\tauS{(s)}(t) = \lockval \quad \tauO{(s)}(t+1) = \unlockval
\]
The equations say that we have locked at time $t$, but another thread
has unlocked at $t+1$.
Modeling such behavior is desirable because there exist locking
algorithms that admit it. For example, even simple spin locks
physically allow that the locking and unlocking threads are different.
Readers-writers locks~\cite{courtois:ACM71}, which can be built over
spin locks, allow an initial reader thread to acquire a lock
and a possibly different reader thread to release it.
%
In a setting where a lock can be released by any thread, one can't
really speak about lock ownership. Thus, structures that admit such
behavior and that can utilize the general histories above, typically
require richer internal ghost state in order to specify the desired
locking discipline. For example, readers-writes locks require
enrichment with permissions~\cite{Bornat-al:POPL05}, which we forego
here.

Nevertheless, even without enrichment, we can already illustrate how
to impose on locking histories a more restricted behavior, whereby the
thread that unlocks must be the one that currently holds the
lock. Such ``mutually exclusive'' histories form a sub-PCM of general
locking histories, and thus the property of mutual exclusion can be
captured as a \seprel.  Analogous to the ticket lock example, we then
construct the morphism $\omega$ that computes lock ownership.

%

Let us first define the separating relation:
\begin{equation}\label{eqn:seprelomega}
\begin{aligned}
	\relOmega x y \eqdef &
  	 (\forall t\ldot \papp{x}{t} = \lockval \Rightarrow
		 \lastKey{x \join y} \leq t \lor \papp{x}{t+1} = \unlockval) \land \hbox{}
	\\
   	& (\forall t\ldot \papp{y}{t} = \lockval \Rightarrow
		 \lastKey{y \join x} \leq t \lor \papp{y}{t+1} = \unlockval) \land
	 x \rel{} y
\end{aligned}
\end{equation}

\noindent Intuitively, the relation states that whenever the thread
with history $x$ locked at time $t$ then the thread with history $y$
couldn't have proceeded. On the other hand, the thread with history
$x$ could have proceeded by unlocking at the immediate time
$t+1$. Similarly to the \seprels in the previous sections, the
relation symmetrically applies to the history $y$ as well, and
requires that the join of $x$ and $y$ be valid, i.e., that the
histories of two threads don't share timestamps.

\begin{lemma}
	The relation $\relOmega {} {}$ is a \seprel.
\end{lemma}
\begin{proof}
  The proof 
  shows that properties of
Definition~\ref{def:seprel} hold:
\begin{enumerate}
\item Definedness: 
  Assume that $\relOmega x y$, and show that $\relOmega x
  \unit$. Indeed, consider $t$ such that $\papp{x}{t} = \lockval$.  By
  case analysis on $t$, either $t = \lastKey{x} = \lastKey{x \join
    \unit}$, or $t \leq \lastKey{s} = \lastKey{x \join \unit}$. In
  either case, trivially $\relOmega x \unit$.

\item Strengthening: Follows trivially form definition of $\relOmega {}
  {}$.

\item Unit property: follows from symmetry (4) and definedness (1).

\item Symmetry: Immediate from symmetry of conjunction and symmetry
  of $\valid x y$.

\item Associativity: Assume $\relOmega {(x \join y)} z$ and
  $\relOmega x y$ to show $\relOmega y z$ and $\relOmega x
  {(y \join z)}$. Let $t$ be a timestamp such that $\papp{y}{t} =
  \lockval$ (the cases when $\papp{x}{t} = \lockval$ or $\papp{z}{t} =
  \lockval$ are similar). Then from $\relOmega x y$ we get that
  $\lastKey{x \join y} \leq t$ (and more specifically $\lastKey{x} <
  t$), or $\papp{y}{t+1} = \unlockval$. In the first case, it must be
  $\papp{(x \join y)}{t} = \lockval$. Thus from $\relOmega {(x \join
    y)} z$, we infer that either $\lastKey{z} < t$, and thus
  $\relOmega y z$ and $\relOmega x (y \join z)$, or $\papp{(x \join
    y)}{t+1} = \unlockval$, which implies that $\papp{y}{t+1} =
  \unlockval$, which we consider as part of the second case. In the
  second case, i.e., when $\papp{y}{t+1} = \unlockval$, the property
  $\relOmega y z$ is immediate. On the other hand, we also have
  $\papp{(y \join z)}{t} = \lockval$ and $\papp{(y \join z)}{t+1} =
  \unlockval$, thus $\relOmega x {(y \join z)}$ holds as well.

\end{enumerate}
\end{proof}

Finally, we can define the morphism $\omega : \Hist \to O$.
%
%
%
\begin{equation}\label{eqn:morphomega}
\fapp{\omega}{h} \eqdef \begin{cases}
  \top & \text{if } h = \top \\
  \own & \text{if } $\lastKey{h} > 0$\ \mbox{and}\ \papp{h}{\lastKey{h}} = \lockval \\
  \nown & \text{otherwise}
  \end{cases}
\end{equation}

\begin{lemma}
The map $\omega$ is a morphism with \seprel $\relOmega {} {}$.
\end{lemma}
\begin{proof}
The properties of Definition~\ref{def:pcmmorph} hold as follows:
\begin{enumerate}
\item Map $\omega$ clearly preserves unit since
 $\lastKey{\unit_\Hist} = 0$; thus $\papp{\omega}{\unit_{\Hist}} = \nown = \unit_{O}$.

\item Undefinedness is preserved trivially.

\item To show distributivity, assume that $\relOmega x y$ and let $t =
  \lastKey{x \join y}$. We consider only the interesting case when $t
  > 0$, and w.l.o.g., $t \in x$, and $\papp{x}{t} = \lockval$.  Then
  by definition of $\omega$, $\fapp{\omega}{x} =\own$. But it must
  also be that $\fapp{\omega}{y} = \nown$, for if otherwise, then by
  $\relOmega x y$, the history $y$ must have an unlocking entry at
  time $t+1$ and thus contains a timestamp beyond $t = \lastKey{x
    \join y}$. Therefore $\papp{\omega}{x \join y} = \own =
  \fapp{\omega}{x} \join \fapp{\omega}{y}$.

\end{enumerate}
\end{proof}

\section{Invertible Morphisms and \SepRels}
\label{s:inversion}

\subsection{Invertibility of Morphisms}
\label{s:invmorph}

As we have seen in the previous sections, the key property of a
morphism $\phi : A \rightarrow B$ is $\phi$ distributes over
$\join$. In other words, if the \emph{argument} of $\phi$ splits into
$a_1 \join a_2$, then the result splits as well, that is:
\[
  \fapp{\phi}{(a_1 \join a_2)} = \fapp{\phi}{a_1} \join \fapp{\phi}{a_2}
\]
under a suitable condition on $a_1$ and $a_2$ expressed as a \seprel
$a_1 \rel{\phi} a_2$.

In verification practice, however, we often have to show the converse:
that if the \emph{result} of $\phi$ is defined
and splits into $b_1 \join b_2$, then the argument must split as well,
that is:
\begin{align}
\label{eq:invmorph}
a \rel{\phi} \unit \land \fapp{\phi}{a} = b_1\join b_2 \Rightarrow \exists a_1\ a_2\ldot a = a_1 \join
a_2 \wedge a_1 \orth_\phi a_2 \wedge \fapp{\phi}{a_1} = b_1 \wedge \fapp{\phi}{a_2} = b_2
\end{align}
We call this property \emph{invertibility}, because it can be seen as
imposing a form of distributivity on the \emph{inverse image}
$\phi^{-1} : \mathcal{P}(B) \rightarrow \mathcal{P}(A)$, where we take
only inverses that are \separate from $\unit$ by $\rel{\phi}$, i.e.
$\phi^{-1} (X) = \{x \in A \mid \fapp{\phi}{x} \in B \wedge x
\rel{\phi} \unit\}$.
Indeed, property (\ref{eq:invmorph}) can be restated compactly as
\begin{align}
\label{eq:finvmorph}
\fapp{\phi^{-1}}{\{b_1 \join b_2\}} \subseteq (\fapp{\phi^{-1}}\{b_1\}) \join_{\orth_\phi} (\fapp{\phi^{-1}}\{b_2\})
\end{align}
where $\join_{\orth_\phi}$ is a special case of the more general
operation $\join_R$ that \emph{lifts} a \seprel $R$ on $A$ to an
operation on sets
$X_1, X_2 \in \mathcal{P}(A)$ as follows.
\[
X_1 \join_R X_2 = \{ a_1 \join a_2 \mid a_1 \in X_1 \wedge a_2 \in X_2 \wedge a_1\,R\,a_2\}
\]

Invertibility of morphisms appears naturally in separation logic when
reasoning by framing or parallel composition. For example, imagine a program
$e$ with the following spec, similar to our abstract spec for $\progLock$.
\[
e : \spec{\lambda s\ldot\phiS{s} = b_1} 
    \spec{\lambda s\ldot \phiS{s} = b'_1}
\]
Here $\phi$ is a total morphism (i.e., $\phi$ has a trivial
\separating relation), and
we want to frame it by $\frameC{(\lambda s\ldot\phiS{s} = b_2)}$. The
direct application of the frame rule, unfolding the definition of
$\bstar$ that we introduced in Section~\ref{s:example}, derives
\begin{align*}
e :~ & \spec{\lambda s\ldot \exists s_1\, s_2\ldot 
            s = s_1 \bstar \frameC{s_2} \land
	    \phiS{s_1} = b_1 \land
            \frameC{\phiS{s_2} = b_2}}\\
     & \spec{\lambda s\ldot \exists s_1\, s_2\ldot 
            s = s_1 \bstar \frameC{s_2} \land
	    \phiS{s_1} = b'_1 \land
            \frameC{\phiS{s_2} = b_2}}
\end{align*}
Of course, we would like to strengthen the precondition and weaken the
postcondition of this spec into the more compact and ultimately
desirable form
\[
e : \spec{\lambda s\ldot \phiS{s} = b_1 \join \frameC{b_2}}
    \spec{\lambda s\ldot \phiS{s} = b'_1 \join \frameC{b_2}}
\]
Here's where invertibility comes in. It's easy to see that the
postcondition readily weakens into the desired form just by using that
$\phi$ is a (total) morphism, and the fact that
$\cs{s} = \cs{s_1} \join \cs{s_2}$. However, the precondition doesn't
strengthen immediately. We need to show
\[
\phiS{s} = b_1 \join b_2 \Rightarrow \exists s_1\ s_2\ldot 
            s = s_1 \bstar s_2 \wedge \phiS{s_1} = b_1 \wedge 
            \phiS{s_2} = b_2
\]
but this doesn't follow from distributivity of $\phi$. It does follow,
however, if $\phi$ is invertible. To see this, assume that
$\cs{s} = a$. Then $\phiS{s} = b_1 \join b_2$ transforms into
$\fapp{\phi}{a} = b_1 \join b_2$. From the assumption that
$\cs{s} \orth \co{s}$ and the properties of \seprels, we get
$\cs{s} \orth \unit$ and thus $a \orth \unit$ as well.  Then
invertibility of $\phi$ gives us $a_1$ and $a_2$ such that
$a = a_1 \join a_2$, $a_1 \orth_\phi a_2$ (which equals
$a_1 \orth a_2$ because $\phi$ is a total morphism),
$\fapp{\phi}{a_1} = b_1$, and $\fapp{\phi}{a_2} = b_2$.  Choosing
$s_1 = (a_1, a_2 \join \co{s})$ and $s_2 = (a_2, a_1 \join \co{s})$
gives us $s = s_1 \bstar s_2$ such that
$\phiS{s_1} = \fapp{\phi}{a_1} = b_1$ and
$\phiS{s_2} = \fapp{\phi}{a_2} = b_2$. This strengthens the
precondition as desired.

\subsection{Invertibility of \SepRels}
\label{s:invseprel}

Similar style of reasoning applies if $\phi$ isn't total, but has a
non-trivial \seprel $\orth_\phi$. It turns out, however, that then we
need to impose an additional condition of $\orth_\phi$, thus giving
rise to a notion of \emph{invertible} \seprels also.
To see what this condition should be, imagine that we have a program
$e$ with the following spec, similar to our intermediate abstract spec
for $\progLock$.
\[
e : \spec{\lambda s\ldot\phiS{s} = b_1 \wedge \cs{s} \rel{\phi} \co{s}} 
    \spec{\lambda s\ldot \phiS{s} = b'_1 \wedge \cs{s} \rel{\phi} \co{s}}
\]
Because $\phi$ is not total, we include the conjunct
$\cs{s} \rel{\phi} \co{s}$ into the spec to ensure that $\phi$
distributes when framed. We now want to frame with
$\frameC{(\lambda s\ldot \phiS{s}=b_2 \wedge \cs{s} \rel{\phi} \co{s})}$.
Similarly to the previous Section~\ref{s:invmorph}, unfolding the definition of
$\star$ derives us the following spec:
\begin{align}\label{eq:framedspec}
e :~ & \spec{\lambda s\ldot \exists s_1\, s_2\ldot 
            s = s_1 \bstar \frameC{s_2} \land
	    \phiS{s_1} = b_1 \land
	    \cs{s_1} \rel{\phi} \co{s_1} \land
            \frameC{\phiS{s_2} = b_2 \land
	    \cs{s_2} \rel{\phi} \co{s_2}}
	  }\\
     & \spec{\lambda s\ldot \exists s_1\, s_2\ldot \notag
            s = s_1 \bstar \frameC{s_2} \land
	    \phiS{s_1} = b'_1 \land
	    \cs{s_1} \rel{\phi} \co{s_1} \land
            \frameC{\phiS{s_2} = b_2 \land
	    \cs{s_2} \rel{\phi} \co{s_2}}
	 }
\end{align}
However, we ultimately desire to obtain a compact spec in the following form:
\begin{align}\label{eq:compactspec}
e : \spec{\lambda s\ldot
	\phiS{s} = b_1 \join \frameC{b_2} \land 
	\cs{s} \rel{\phi} \co{s}} 
    \spec{\lambda s\ldot
    	\phiS{s} = b'_1 \join \frameC{b_2}  \land 
    	\cs{s} \rel{\phi} \co{s}}
\end{align}
As before, we need to prove two implications to weaken (\ref{eq:framedspec})
to (\ref{eq:compactspec}).
%
%
\begin{align}
& \phiS{s} = b_1 \join b_2 \land \cs{s} \rel{\phi}
                \co{s} \Rightarrow \hbox{} \label{eq:inv1} \\
& \qquad \begin{aligned}[t]
  \exists s_1\, s_2\ldot s = s_1 * s_2 \land 
                          \phiS{s_1} = b_1 \land \cs{s_1} \rel{\phi} \co{s_1} \land 
                          \phiS{s_2} = b_2 \land \cs{s_2} \rel{\phi} \co{s_2} \notag
  \end{aligned}\\
&  s = s_1 * s_2 \land 
      \phiS{s_1} = b'_1 \land \cs{s_1} \rel{\phi} \co{s_1} \land 
      \phiS{s_2} = b_2 \land \cs{s_2} \rel{\phi} \co{s_2} 
  \Rightarrow \hbox{} \label{eq:inv2} \\
& \qquad \phiS{s} = b'_1 \join b_2 \land \cs{s} \rel{\phi} \co{s} \notag
\end{align}
Or alternatively, if we replace the state variables by pairs of
\emph{self} and \emph{other} components, e.g., $s = (a, a')$,
$s_1 = (a_1, a_2 \join a')$, $s_2 = (a_2, a_1 \join a')$, 
%
we obtain after some simplification:
\begin{align}
& 
    \fapp{\phi}{a} = b_1 \join b_2 \land a
                \rel{\phi} a' \Rightarrow \hbox{} \label{eq:inv3}\\
& \qquad \begin{aligned}[t]
  \exists a_1\, a_2\ldot a = a_1 \join a_2 \land 
                         \fapp{\phi}{a_1} = b_1 \land a_1 \rel{\phi} (a_2 \join a') \land 
                         \fapp{\phi}{a_2} = b_2 \land a_2 \rel{\phi}
                         (a_1 \join a') \notag
  \end{aligned}\\
&  \fapp{\phi}{a_1} = b'_1 \land a_1 \rel{\phi} (a_2 \join a') \land 
      \fapp{\phi}{a_2} = b_2 \land a_2 \rel{\phi}  (a_1 \join a') 
  \Rightarrow \hbox{} \label{eq:inv4}\\
& \qquad \fapp{\phi}{(a_1 \join a_2)} = b'_1 \join b_2 \land (a_1
                                          \join a_2) \rel{\phi} a' \notag
\end{align}
If we assume that $\phi$ is invertible, then from $a \rel{\phi} a'$,
we get $a \rel{\phi} \unit$ by the definedness property of \seprels,
and then (\ref{eq:inv3}) immediately follows by associativity of
\seprels.  However, to obtain the first conjunct in the conclusion of
(\ref{eq:inv4}),
we require that $a_1 \rel{\phi} a_2$, so that we can distribute $\phi$
over $a_1 \join a_2$ and then use that $\fapp{\phi}{a_1} = b_1$ and
$\fapp{\phi}{a_2} = b_2$. To obtain the second conjunct in
(\ref{eq:inv4}), we need to reassociate $a_1$, $a_2$ and $a'$, which
can be done if $a_1 \rel{\phi} a_2 \rel{\phi} a'$.
Thus, we obtain the
required condition that makes it possible to derive (\ref{eq:compactspec}).
\[
a_1 \rel{\phi} (a_2 \join a') \wedge a_2 \rel{\phi} (a_1 \join a')
\Rightarrow a_1 \rel{\phi} a_2 \rel{\phi} a'
\]
To establish this implication it suffices to show that either $a_1 \rel{\phi} a'$ or $a_2 \rel{\phi} a'$ as 
the consequent $a_1 \rel{\phi} a_2 \rel{\phi} a'$ then follows from
associativity of \seprels.


\subsection{Duality of Invertibility of Morphisms and \SepRels}

We note an interesting duality in the interplay of $\phi$ and
$\rel{\phi}$ in the above framing process. When strengthening the
precondition, it's the invertibility of $\phi$ that provides the split
of $a$ into $a_1 \join a_2$ such that $a_1 \rel{\phi} a_2$, which is
then used to reassociate $\rel{\phi}$. When weakening the
postcondition, the situation is dual. We start with $a$ already split
into $a = a_1 \join a_2$, but it's the invertibility of $\rel{\phi}$
that ensures the split is such that $\phi$ can distribute over it.
Thus, in the precondition, $\phi$ helps $\rel{\phi}$ and in the
postcondition $\rel{\phi}$ helps $\phi$.

Thus, to summarize, we have the following definitions of invertibility
for \seprels and morphisms that enable framing in the abstract of
specs of above form, i.e., without relying on the definitions of
morphism or its \seprel.

\begin{definition}
\label{d:invseprel}
A \seprel $R$ on the PCM $A$ is \emph{invertible} if for all $a_1$, 
$a_2$, $a'$ such that $a_1 ~R~ (a_2 \join a')$ and 
$a_2 ~R~ (a_1 \join a')$, it must also be $a_1 ~R~ a_2 ~R~
a'$. Moreover, it suffices to prove $a_1~R~a'$ or $a_2~R~a'$, as $a_1~R~a_2~R~a'$ follows by associativity. 

\end{definition}

\begin{definition}
\label{d:invmorph}
A morphism $\phi : A \to B$ is \emph{invertible} if $\rel{\phi}$ is an
invertible \seprel and for all $a \in A$ such that
$a \rel{\phi} \unit$, and $b_1, b_2 \in B$ where
$\fapp{\phi}{a} = b_1 \join b_2$, there exist $a_1, a_2 \in A$, such
that $a = a_1 \join a_2$, $a_1 \rel{\phi} a_2$,
$\fapp{\phi}{a_1} = b_1$ and $\fapp{\phi}{a_2} = b_2$.
\end{definition}


We 
now demonstrate the invertibility of various constructions we
introduced earlier. First, \seprels of total morphisms are always
invertible.

\begin{proposition}
	\label{p:trivinv}
	Let $A$ be a PCM. The trivial separating relation $\rel{A}$ is
	invertible.
\end{proposition}
\begin{proof}
	Let $a_1 \rel{A} (a_2 \join a')$. Recall that the trivial \seprel is
	given by $x \rel{A} y \eqdef (x \join y) \text{ is defined}$. 
	Hence we obtain that $a_1 \join {(a_2 \join a')}$ is defined, and, using
	commutativity and associativity of join $\join$, we have that $(a_1 \join a') \join
	a_2$ is defined. Thus, by law (3) of Definition \ref{def:TPCM}, also $a_1 \join a'$ is
	defined whence $a_1 \rel{A} a'$. Similarly for $a_2 \rel{A} a'$.
\end{proof}

\noindent Similarly, other basic constructions on \seprel preserve
invertibility. So do the construction on morphisms.
%
%
%
Recall the composition of morphisms, tensor and arrow product in
Definition~\ref{def:constrs}.

\begin{proposition}
\label{p:invcomptimes}
Let $\alpha$, $\beta$ be invertible morphisms. Then
$\comp{\alpha}{\beta}$ and 
$\alpha \times \beta$ 
are invertible morphisms.
\end{proposition}


\begin{proof}
  We just show the case for $\circ$ as the one for $\times$ is simple.
  Let $\alpha : C \to A$, $\beta : B \to C$ be invertible morphisms.
First, we show that $\rel{\comp{\alpha}{\beta}}$, the induced \seprel of $\comp{\alpha}{\beta}$, is invertible. 
Consider
$a_1$, $a_2$, $a'$,
such that
$a_1 \rel{\comp{\alpha}{\beta}} (a_2 \join a')$, and
$a_2 \rel{\comp{\alpha}{\beta}} (a_1 \join a')$. 
We need to show that $a_1 \rel{\comp{\alpha}{\beta}}
a_2 \rel{\comp{\alpha}{\beta}} a'$.
From Definition~\ref{def:constrs} of composition, we obtain
$a_1 \rel{\beta} (a_2 \join a') \land \fapp{\beta}{a_1} \rel{\alpha} \papp{\beta}{a_2 \join a'}$, and
$a_2 \rel{\beta} (a_1 \join a') \land \fapp{\beta}{a_2} \rel{\alpha} \papp{\beta}{a_1 \join a'}$.
%
%
Since $\beta$ is an invertible morphism also its \seprel $\rel{\beta}$ is
invertible. We use invertibility of $\rel{\beta}$ and the first conjunct to
obtain that
$a_1 \rel{\beta} a_2 \rel{\beta} a'$.
Now considering the second conjuncts, we get by distributivity of $\beta$,
$\fapp{\beta}{a_1} \rel{\alpha} (\papp{\beta}{a_2} \join \papp{\beta}{a'})$ and 
$\fapp{\beta}{a_2} \rel{\alpha} (\papp{\beta}{a_1} \join \papp{\beta}{a'})$.
Because $\alpha$ is invertible, so is $\rel{\alpha}$. We therefore obtain
$\fapp{\beta}{a_1} \rel{\alpha} \fapp{\beta}{a_2} \rel{\alpha} \fapp{\beta}{a'}$.
Thus $\rel{\comp{\alpha}{\beta}}$ is invertible.

Second, we show that $\comp{\alpha}{\beta}$ is an invertible morphism. 
Assume 
$b \in B$ and $a_1,a_2 \in A$ 
such that
$\fapp{(\comp{\alpha}{\beta})}{b} =
a_1 \join a_2$. Also assume that $b \rel{\comp{\alpha}{\beta}} \unit$;
that is, $b \rel{\beta} \unit$ and
$\fapp{\beta}{b} \rel{\alpha} \unit$.
Using invertibility of $\alpha$ on $\fapp{\beta}{b} \in C$,
we obtain
$c_1,c_2 \in C$, such that
$\fapp{\beta}{b} = c_1 \join c_2$,
$c_1 \rel{\alpha} c_2$,
$\fapp{\alpha}{c_1} = a_1$, and
$\fapp{\alpha}{c_2} = a_2$.
Using invertibility of $\beta $ on $b$ we further obtain
$b_1, b_2 \in B$, such that
$b = b_1 \join b_2$,
$b_1 \rel{\beta} b_2$,
$\fapp{\beta}{b_1} = c_1$, and
$\fapp{\beta}{b_2} = c_2$.
%
Consequently 
$\fapp{\beta}{b_1} \rel{\alpha} \fapp{\beta}{b_2}$.
Hence, using $b_1\rel{\beta}b_2$, 
we obtain $b_1 \rel{\comp{\alpha}{\beta}} b_2$.
Finally,
$\fapp{(\comp{\alpha}{\beta})}{b_1} = \papp{\alpha}{\fapp{\beta}{b_1}} =
\papp{\alpha}{c_1} = a_1$ and similarly for $a_2$. 
Therefore, we have
$b_1, b_2 \in B$ such that 
$b = b_1 \join b_2$,
$b_1 \rel{\comp{\alpha}{\beta}} b_2$,
$\fapp{(\comp{\alpha}{\beta})}{b_1} = a_1$ and
$\fapp{(\comp{\alpha}{\beta})}{b_2} = a_2.$
Hence morphism $\comp{\alpha}{\beta}$ is invertible.
\end{proof}

Notice that $\alpha \otimes \beta$ is an example of a morphism that
isn't necessarily invertible, even if $\alpha$ and $\beta$ are.  By
definition,
$\fapp{(\alpha \otimes \beta)}{x} = (\fapp{\alpha}{x},
\fapp{\beta}{x})$. Thus, if we're given
$\fapp{(\alpha \otimes \beta)}{x} = (y, z)$, we can induce one split
of $x$ by $\alpha$ and $y$, and another by $\beta$ and $z$. However,
there's no reason to expect that these splits are equal, which is
required for $\alpha \otimes \beta$ to be invertible.


We also introduced the notions of kernel and equalizer, which are
\seprels. These illustrate constructions that turn invertible
morphisms into invertible \seprels.

\begin{proposition} 
  \label{p:invkereqlz}
  Let $\alpha$, $\beta$ be morphisms with invertible \seprels. Then
  $\fapp{\ker}{\alpha}$ is an invertible \seprel, while
  $\fapp{\fapp{\eql}{\alpha}}{\beta}$ is so if the range PCM of
  $\alpha$ and $\beta$ is cancellative.\footnote{A PCM is cancellative
    if $a \join b = a \join c$ implies $b = c$, whenever $a \orth b$
    and $a \orth c$.}
\end{proposition}


\begin{proof}
  We show the proof for equalizers. A kernel is a special case of an
  equalizer when one of the morphisms is the always-unit one, which
  circumvents the need for cancellativity.
  Consider $x$, $y$, $z$, such that
  $x\,(\fapp{\fapp{\eql}{\alpha}}{\beta})\,(y \join z)$ and
  $y\,(\fapp{\fapp{\eql}{\alpha}}{\beta})\,(x \join z)$. By
  Definition~\ref{d:invseprel}, it suffices to show
  $y\,(\fapp{\fapp{\eql}{\alpha}}{\beta})\,z$; that is
  $y \rel{\alpha} z \wedge y \rel{\beta} z \wedge \fapp{\alpha}{y} =
  \fapp{\beta}{y} \wedge \fapp{\alpha}{z} = \fapp{\beta}{z}$.
  From the assumptions, we get
  $x \rel{\alpha} (y \join z) \land x \rel{\beta} (y \join z) \wedge
  \fapp{\alpha}{x} = \fapp{\beta}{x} \wedge \papp{\alpha}{y \join z} =
  \papp{\beta}{y \join z}$, and
  $y \rel{\alpha} (x \join z) \land y \rel{\beta} (x \join z) \wedge
  \fapp{\alpha}{y} = \fapp{\beta}{y} \wedge \papp{\alpha}{x \join z} =
  \papp{\beta}{x \join z}$.
  Since $\rel{\alpha}$ and $\rel{\beta}$ are both invertible \seprels,
  this obtains $y \rel{\alpha} z \wedge y \rel{\beta} z$.
  Thus, we can distribute $\alpha$ and $\beta$ over $y \join z$ to
  derive:
  $\fapp{\alpha}{y} \join \fapp{\alpha}{z} = \fapp{\beta}{y} \join
  \fapp{\beta}{z}$. Since we already have
  $\fapp{\alpha}{y} = \fapp{\beta}{y}$, we apply cancellativity to
  derive $\fapp{\alpha}{z} = \fapp{\beta}{z}$ and conclude the proof.
\end{proof}

\subsubsection*{Sub-PCM}

Section~\ref{s:examplesubpcm} demonstrates how to use the sub-PCM
construction to provide a compact spec. First, we start with a spec like the
following:
\[
e : \spec{\lambda s \ldot \alphaS{s} = b_1 \land \relAlpha {\cs{s}} {\co{s}}} 
    \spec{\lambda s \ldot \alphaS{s} = b'_1 \land \relAlpha {\cs{s}} {\co{s}}} 
\]
Using the sub-PCM construction, we can write the spec compactly 
as follows:
\[
e : \spec{\lambda s \ldot \fapp{\cs{(\comp{\alpha}{\iota})}}{s} = b_1}
    \spec{\lambda s \ldot \fapp{\cs{(\comp{\alpha}{\iota})}}{s} = b'_1}
\]
Note that, implicitly, we also have that $\valid {\cs{s}} {\co{s}}$. 
We show the following theorem that states that invertibility is preserved by
such construction:
\begin{theorem}
\label{thm:iota}
Let $\alpha : A \to B$ be an invertible morphism and let
$\iota : A/{\rel{\alpha}} \to A$ be a sub-PCM injection. Then
$\comp{\alpha}{\iota} : A/{\rel{\alpha}} \to B$ is invertible.
\end{theorem}
\begin{proof}
  Recall the sub-PCM construction in Section~\ref{s:subpcm} and use
  $\rel{\alpha}$ as the \seprel for the construction of sub-PCM.  Then
  $\fapp{\iota}{x} \eqdef x$ and both $\iota$ and
  $\comp{\alpha}{\iota}$ are total morphisms, with the \seprel
  $\rel{A/\rel{\alpha}}$. This \seprel is a restriction of
  $\rel{\alpha}$ to the set
  $A/{\rel{\alpha}} = \{a \in A \mid a \rel{\alpha} \unit_A\}$. The
  \seprel is also trivial and thus invertible, by
  Proposition~\ref{p:trivinv}.
%
%

Now we proceed with the proof of invertibility itself.  Assume that we
are given $a \in A/{\rel{\alpha}}$ and $b_1, b_2 \in B$ such that
$\fapp{(\comp{\alpha}{\iota})}{a} = \papp{\alpha}{\fapp{\iota}{a}} =
\fapp{\alpha}{a} = b_1 \join b_2$ and $a \rel{A/{\rel{\alpha}}} \unit$.
The second conjunct implies $a \rel{\alpha} \unit_A$.
Now, because $\alpha$ is invertible, there exist $a_1, a_2 \in A$ such
that $a = a_1 \join_A a_2$, $\relAlpha a_1 a_2$,
$\fapp{\alpha}{a_1} = b_1$, and $\fapp{\alpha}{a_2} = b_2$.
But, because $a_1 \rel{\alpha} a_2$ it follows that
$a_1, a_2 \in A/{\rel{\alpha}}$,
$a = a_1 \join_{A/{\rel{\alpha}}} a_2$, and
$a_1 \rel{A /{\rel{\alpha}}} a_2$.  Since also
$\fapp{(\comp{\alpha}{\iota})}{a_i} = \fapp{\alpha}{a_i} = b_i$, the
morphism $\comp{\alpha}{\iota}$ is invertible.
\end{proof}

We can now show that the morphisms and \seprels used in our abstract
specs (both the intermediate and final one) of ticket lock are invertible. 

\begin{lemma}
\label{lem:alphaseprel}
 The \seprel $\rel{\alpha}$ from (\ref{eqn:seprelalpha}) is invertible. 
\end{lemma}
\begin{proof}
Assume that there are
$a_1, a_2, a'$ such that
$\relAlpha a_1 {(a_2 \join a')}$ and
$\relAlpha a_2 {(a_1 \join a')}$.
Using the definition of $\rel{\alpha}$, we obtain
$\numCurrent{\fapp{\sigma}{a_1}} + \numCurrent{\fapp{\sigma}{(a_2 \join a')}}
\leq 1 \land a_1 \orth (a_2 \join a')$ and
$\numCurrent{\fapp{\sigma}{a_2}} + \numCurrent{\fapp{\sigma}{(a_2 \join a')}}
\leq 1 \land a_2 \orth (a_1 \join a')$.
This gives us,
using the second conjuncts and commutativity and associativity of $\join$ as
in the proof of Proposition~\ref{p:trivinv},
that $a_1 \orth a_2 \orth a'$. Further, since $\numCurrent{-}$ and
$\sigma$ are morphisms, we obtain, using either of the first conjuncts,
$\numCurrent{\fapp{\sigma}{a_1}} + 
\numCurrent{\fapp{\sigma}{a_2}}  + 
\numCurrent{\fapp{\sigma}{a'}}
\leq 1$.
Therefore 
$\numCurrent{\fapp{\sigma}{a_1}} + 
\numCurrent{\fapp{\sigma}{a'}}
\leq 1$ 
and we conclude that
$\relAlpha a_1 a'$.
The rest follows from associativity of \seprels.
\end{proof}

\begin{lemma}
\label{lem:alphaprime}
 The morphism $\alpha$ from (\ref{eqn:morphalpha}) is invertible. 
\end{lemma}
\begin{proof}
The \seprel $\rel{\alpha}$ is invertible by Lemma~\ref{lem:alphaseprel}.
Now, assume that there are $a \in U$ and $b_1, b_2 \in O$ such that
$\fapp{\alpha}{a} = b_1 \join b_2$ and $a \rel{\alpha} \unit$.
We must show there exist $a_1, a_2\in U$, such that $a=a_1\join a_2$,
$\relAlpha a_1 a_2$, $\alpha\;a_1=b_1$, and $\alpha\;a_2=b_2$.
Proceed by case analysis on $b_1 \join b_2$.

\emph{Case:} $b_1 \join b_2 = \own$. 
W.l.o.g. $b_1 = \own$ and $b_2 = \nown$. Choose $a_1=a$, $a_2=\unit$.
Then trivially $a=a_1\join a_2$ and by assumption $\relAlpha a_1
a_2$. Also, $\fapp{\alpha}{a_1} = \fapp{\alpha}{a} = \own = b_1$ and
$\fapp{\alpha}{a_2} = \fapp{\alpha}{\unit} = \unit_O = \nown = b_2$.

\emph{Case:} $b_1 \join b_2 = \nown$. Then
$b_1 = b_2 = \nown$. Choose $a_1=a$, $a_2=\unit$. Again trivially
$a=a_1\join a_2$ and by assumption $\relAlpha a_1 a_2$. Also,
$\fapp{\alpha}{a_1} = \fapp{\alpha}{a} = \nown = b_1$ and 
$\fapp{\alpha}{a_2} = \fapp{\alpha}{\unit} = \unit_O = \nown = b_2$.
\end{proof}

\begin{corollary}
  The morphism $\alpha' = \alpha \join \iota_{\rTicket}$ from
  Section~\ref{s:examplesubpcm} is invertible.
\end{corollary}


The same holds also for the morphism and the separating relation we discussed
in Section~\ref{sec:spinex}:

\begin{lemma}
\label{lem:omegaseprel}
 The \seprel $\rel{\omega}$ from (\ref{eqn:seprelomega}) is invertible. 
\end{lemma}
\begin{proof}
  Assuming $\relOmega a_1 {(a_2 \join a')}$ and
  $\relOmega a_2 {(a_1 \join a')}$, by associativity of
  $\relOmega {} {}$, it suffices to establish $\relOmega a_1 a'$. In
  Definition~\ref{d:invseprel}, we only consider the clause whereby
  $\papp{a'}{t} = \lockval$ implies $\lastKey{a_1 \join a'} \leq t$ or
  $\papp{a'}{t+1} = \unlockval$. From $\papp{a'}{t} = \lockval$, it
  follows that $\papp{(a_2 \join a')}{t} = \lockval$. Therefore,
  $\relOmega a_1 {(a_2 \join a')}$ derives that
  $\lastKey{a_1 \join (a_2 \join a')} \leq t$ or
  $\papp{(a_2 \join a')}{t+1} = \unlockval$. In the first case, it
  must also be
  $\lastKey{a_1 \join a'} \leq \lastKey {a_1 \join (a_2 \join a')}
  \leq t$, which completes the proof. In the second case, it can be
  either $\papp{a'}{t+1} = \unlockval$ or
  $\papp{a_2}{t+1} = \unlockval$. The first case also completes the
  proof. The second case contradicts the assumption
  $\relOmega a_2 {(a_1 \join a')}$, and is thus impossible.
%
\end{proof}

\begin{lemma}
\label{lem:omega}
 The morphism $\omega$ from (\ref{eqn:morphomega}) is invertible. 
\end{lemma}
\begin{proof}
The \seprel $\rel{\omega}$ is invertible by Lemma~\ref{lem:omegaseprel}.
The rest of the proof follows similarly as in the case of
Lemma~\ref{lem:alphaprime}.
\end{proof}
%



\subsection{Invertibility and Separating Conjunction}
\label{s:cap}

We close this section with two lemmas
that show how invertible
morphisms and \seprels interact with separating conjunction. We'll
elaborate more on these properties in Section~\ref{s:relwork} on the
related work.

\begin{lemma}\label{p:duplicable}
  Let $S$ be an invertible \seprel, and let 
  $R = \lambda s\ldot (\cs{s}~S~\co{s})$. Then $R$ is duplicable, 
  i.e., $R \Leftrightarrow R \bstar R$. 
\end{lemma}
\begin{proof}
   For the $\Rightarrow$ direction, let's assume that $s = (a, a')$
   and $\fapp{R}{s}$; that is $a~S~a'$. 
   Consider states $s_1 = s = (a, a')$ and $s_2 = (\unit, a \join a')$. 
   By definition, $s = s_1 \star s_2$. 
   For $s_1$, we do have $\cs{s_1}~S~\co{s_1}$. Indeed, the latter by 
   definition equals $a~S~a'$, and thus holds by assumption.  For 
   $s_2$, we do have $\cs{s_2}~S~\co{s_2}$. Indeed, the latter by 
   definition equals $\unit~S~(a \join a')$, which holds by 
   Proposition~\ref{p:sepU0}. But then $\fapp{R}{s_1}$ and 
   $\fapp{R}{s_2}$, and thus $\papp{(R \bstar R)}{s}$. 

   For the $\Leftarrow$ direction, let's assume $s = s_1 \star s_2$
   where $s = (a_1 \join a_2, a')$, $s_1 = (a_1, a_2 \join a')$ and 
   $s_2 = (a_2, a_1 \join a')$, such that $\fapp{R}{s_1}$ and 
   $\fapp{R}{s_2}$. 
   That is, for $s_1$: $a_1~S~(a_2 \join a')$. And for $s_2$: 
   $a_2~S~(a_1 \join a')$. By invertibility of $S$
   then $a_1~S~a_2~S~a'$, and thus by associativity 
   $(a_1 \join a_2)~S~a'$, i.e. $\fapp{R}{s}$. 
\end{proof}

\begin{lemma}\label{p:nonduplicable}
  Let $\phi$ be an invertible morphism, and let 
  $\fapp{F}(b) = \lambda s\ldot (\phiS{s} = b \wedge \cs{s} \rel{\phi}
  \co{s})$. Then 
  $\fapp{F}{(b_1 \join b_2)} \Leftrightarrow \fapp{F}{b_1} \bstar \fapp{F}{b_2}$.
\end{lemma}
%
%
\begin{proof}
   For the $\Rightarrow$ direction, let $s = (a, a')$ and 
   $\papp{\fapp{F}{(b_1 \join b_2)}}{s}$; that is 
   $\fapp{\phi}{a} = b_1 \join b_2$ and $a \rel{\phi} a'$. 
   By defined elements property of $\rel{\phi}$, it must be 
   $a \rel{\phi} \unit$. Then by invertibility of $\phi$, there exist 
   $a_1$ and $a_2$, such that $a = a_1 \join a_2$, 
   $a_1 \rel{\phi} a_2$, $\fapp{\phi}{a_1} = b_1$ and 
   $\fapp{\phi}{a_2} = b_2$. 
   From $(a_1 \join a_2) \rel{\phi} a'$ and $a_1 \rel{\phi} a_2$, by 
   associativity of \seprels, we get $a_1 \rel{\phi} (a_2 \join a')$
   and $a_2 \rel{\phi} (a_1 \join a')$. Combined with 
   $\fapp{\phi}{a_1} = b_1$ and $\fapp{\phi}{a_2} = b_2$, we get 
   $\papp{\papp{F}{b_1}}{s_1}$ and $\papp{\papp{F}{b_2}}{s_2}$, where 
   $s_1 = (a_1, a_2 \join a')$ and $s_2 = (a_2, a_1 \join a')$. Because 
   also $s = s_1 \star s_2$, we get 
   $\papp{(\papp{F}{b_1} \bstar \papp{F}{b_2})}{s}$. 

   For the $\Leftarrow$ direction, let $s = s_1 \star s_2$ where 
   $s =(a_1 \join a_2, a')$, $s_1 = (a_1, a_2 \join a')$ and 
   $s_2 = (a_2, a_1 \join a')$, such that $\papp{\papp{F}{b_1}}{s_1}$
   and $\papp{\papp{F}{b_2}}{s_2}$. 
   That is, $\fapp{\phi}{a_1} = b_1$ and 
   $a_1 \rel{\phi} (a_2 \join a')$ and $\fapp{\phi}{a_2} = b_2$ and 
   $a_2 \rel{\phi} (a_1 \join a')$. By invertibility of $\rel{\phi}$, 
   then $a_1 \rel{\phi} a_2 \rel{\phi} a'$, and by associativity 
   $(a_1 \join a_2) \rel{\phi} a'$.  
   By distributivity of $\phi$, also 
   $\fapp{\phi}(a_1 \join a_2) = \fapp{\phi}{a_1} \join 
   \fapp{\phi}{a_2} = b_1 \join b_2$. In other words, 
   $\papp{\papp{F}{b_1\join b_2}}{s}$. 
\end{proof}

\label{sex:spinex}




\section{Related Work}
\label{s:relwork}

\paragraph*{PCMs in Separation Logics} 
PCMs arise as the structure underpinning the semantics of (concurrent)
separation logic: the PCMs of heaps capture the dynamics of ownership
transfer which is quintessential to separation logics.
Initially, cancellative PCMs, also known as \emph{separation
  algebras}~\cite{CalcagnoOY07} were used to provide abstract semantic
treatment of separation logic.
Later,~\citet{CaoCA+aplas17} unified different semantics of separation
logics using \emph{ordered} separation algebras to account for
\emph{affine} aspects of various memory models; that is, to model
whether deallocation is explicitly allowed to the user, or is carried
out implicitly by garbage collection. Several program logics continue
this trend, adding further properties to PCMs to give semantics to
(higher-order) ghost
state~\cite{GotsmanBCRS+aplas07,DinsdaleYoung-al:ECOOP10,HoborDA+popl10,Krishnaswami-al:ICFP12,Svendsen-al:ESOP13,DinsdaleYoung-al:POPL13,ArrozPincho-al:ECOOP14,Svendsen-Birkedal:ESOP14,TuronVD+oopsla14,Jung-al:POPL15,jung:jfp18,steelcore+ICFP20}. In
this paper we don't consider higher-order state and focus on the
algebraic treatment of PCMs without additional properties, as these
aren't required by our ambient logic, which admits explicit
deallocation. We expect that in the future morphisms and separating
relations can be developed for these enriched PCMs. 

Recently, several program logics, most notably those that are built on top
of the Iris
framework~\cite{Jung-al:POPL15,jung:jfp18,gra+biz+kre+bir:iron19,HinrichsenBK+popl20,JungLPRTDJ+popl20},
the SteelCore framework~\cite{steelcore+ICFP20}, VST~\cite{Appel-al:BOOK14},
and also
FCSL~\cite{LeyWild-Nanevski:POPL13,Nanevski-al:ESOP14,Sergey-al:ESOP15,Sergey-al:PLDI15,sergey:oopsla16},
have allowed PCMs to be declared at the user level, and sometimes
even constructed by means of a predetermined set of combinators.

However,
%
none of these logics have considered morphisms over PCMs, as we do
here. Instead, when the state space of a program has to be restricted
by some property, that is usually done by conjoining the property to
the state space of the underlying state transition system. In
contrast, with PCM morphisms, we can restrict the PCM itself, thus
promoting the property into a new notion of separateness. The move
makes it possible to provide clients with the PCM most suitable to
their needs. The new PCM may also be subjected to mathematical
theories and their mechanizations that are parametric in the PCM, such
as, for example, our theory of invertibility, to facilitate the reuse
of mechanized proofs.
%

Morphisms are a standard component in the study of structures in
algebra and category theory. They provide the user with the most
general and systematic way to define her own PCM combinators and, as
we illustrated, are also useful in specs. Morphisms generally are also
essential in the definitions of functors and natural transformations
which we plan to consider in the PCM setting in the future. In
contrast to our morphism-based specifications, most of the related
program logics follow the specification style originating from the
work on Concurrent Abstract Predicates
(CAP)~\cite{DinsdaleYoung-al:ECOOP10}, to which we compare below.

We aren't aware of any other work that considers \seprels as a
standalone concept. 
That said, the key \seprel property of associativity (property 5 in
Definition~\ref{def:seprel}) has been considered
before~\cite{Krebbers2015,Jacobs18}, though as a property of the
disjointness relation $\orth$ of the underlying PCM. In our setting,
the latter is just one possible \seprel, associated with total PCM
morphisms.

\paragraph*{Comparison with concurrent abstract predicates (CAP)}
The intermediate abstract specs for ticket locks we developed in
Section~\ref{s:morphabst} are similar to the lock specs from
CAP~\cite{DinsdaleYoung-al:ECOOP10}. We show the
CAP specs below, ignoring lock invariants (see Footnote \ref{ftn:invars} in
Section~\ref{s:example}),
adapted to our type-based notation with explicit binding of the state
$s$ in the assertions, and using $\wedge$ instead of $\bstar$.
\begin{align*}
{\progLock_{CAP}} : 
    {\spec{\lambda s\ldot \fapp{\isalock}{s} \land \fapp{\unlocked}{s}}}
    {\spec{\lambda s\ldot \fapp{\isalock}{s} \land \fapp{\locked}{s}}} \\
{\progUnlock_{CAP}} : 
    {\spec{\lambda s\ldot \fapp{\isalock}{s}\land \fapp{\locked}{s}}}
    {\spec{\lambda s\ldot \fapp{\isalock}{s}\land \fapp{\unlocked}{s}}}
\end{align*}
Here $\isalock$, $\locked$, and $\unlocked$ are separation logic
assertions (hence, predicates over $s$). The predicate $\isalock$
captures the internal conditions required of $s$ to represent a lock,
and $\locked$ and $\unlocked$ capture that the lock is taken and free,
respectively.\footnote{In~\cite{DinsdaleYoung-al:ECOOP10}, the $\unlocked$ predicate
is replaced by separation logic $\mathsf{emp}$, and thus elided. We
include it here explicitly to exemplify the similarity with our
specs.} 
The definitions of the predicates are hidden from the clients, but for
the specs to be usable wrt.~framing and parallel composition, one must
export a number of their properties, such as (a) $\locked \bstar
\locked \Rightarrow \bot$,
and (b) $\isalock$ is duplicable, i.e.  $\isalock \Leftrightarrow
\isalock \bstar \isalock$.

We could turn our specs of $\progLock$ and $\progUnlock$ into the same 
format by setting, for example: 
\begin{align*}
\fapp{\isalock}{s}\,{\eqdef}\,\cs{s}\,{\orth_\alpha}\,\co{s} \quad
\fapp{\unlocked}{s}\,{\eqdef}\,\fapp{\cs{\alpha}}{s} = \nown \wedge
\cs{s}\,{\orth_\alpha}\,\co{s} \quad
\fapp{\locked}{s}\,{\eqdef}\,\fapp{\cs{\alpha}}{s} = \own \wedge \cs{s}\,{\orth_\alpha}\,\co{s}
\end{align*}
and because $\orth_\alpha$ is an invertible \seprel and $\alpha$ an
invertible total morphism, by Lemmas~\ref{p:duplicable} and
\ref{p:nonduplicable}, the equations (a) and (b) above hold.
%
In this sense, we see our algebraic formulation as explaining why the
two different kinds of abstract predicates appear in CAP: the
duplicable predicates are a lifting of invertible \seprels as in
Lemma~\ref{p:duplicable}, and the non-duplicable ones are a lifting of
invertible morphisms as in Lemma~\ref{p:nonduplicable}.


Of course, morphisms and \seprels have uses where abstract predicates
simply don't apply. Examples are the algebraic constructions that we
introduced in Section~\ref{s:morph}, or the sub-PCM construction which
we used to obtain the ultimately simplest abstract specs in
Section~\ref{s:examplesubpcm}.
Furthermore, by being functions, morphisms can compute values out of
the state,\footnote{For example, how we used $\psiT$ in the concrete
  specs for ticket lock to compute the displayed ticket.} and thus
lead to convenient specs and proofs in a formalization based on type
theory. We thus propose that abstract specs 
be given directly 
in terms of morphisms and \seprels, instead of using their coercion
into abstract predicates.

%
%



\paragraph*{Comparison with the ambient type theory}
This paper builds on previous work by~\citet{NanevskiBDF+oopsla19}
which provides a type-theoretic formulation of concurrent separation
logic.
Nanevski et al.~consider an algebraic treatment of \emph{state
  transition systems} of \emph{resources}, introduces notions of
\emph{resource morphisms} and \emph{simulations}. While that paper
focuses on the logic of Hoare triples, in the present paper we focus
on the logic of \emph{assertions} and the associated algebraic
constructions.

The goal of Nanevski et al.~is to provide a systematic way of coercing
a program from one resource type to another, as long as the target
resource simulates the source one. The system provides an inference
rule in the style of Hoare's rule of invariance, to reason about the
coerced programs. We utilized this rule implicitly in
Section~\ref{s:examplesubpcm} to coerce $\progLock$ from a resource
with PCM $U$ to one with PCM $\UTicket$.
A program is coerced from resource $V$ to resource $W$ by means of a
resource morphism, which modifies the behavior of the program on the
ghost state. Programmatically, the coercion may be seen as
\emph{re-instrumenting} a program with a ghost code specific to $W$, a
posteriori to the proof of the program against the initial ghost
intrumentation specific to $V$, and using the resources as a
type-style interface.
The same mechanism of resource morphisms provides a scoped way to
allocate a new resource into the private state of another resource.
Resource morphisms are similar in spirit to the refinement mappings
of~\citet{aba+lam:91}, and enable a form of refinement-style reasoning
within separation logic.

%

%

\paragraph*{PCM morphisms versus homomorphisms in effect algebras.}
\emph{Effectus theory}~\cite{Effectus15} is a fairly new field of
category theory whose aim is to describe quantum computation and its
logic, hence generalizing probabilistic and Boolean logic. The
mathematical backbone of effectus theory is effect algebras, which
essentially are PCMs with an orthosupplement, i.e. a total unary
negation operation.
An effectus is a category with finite coproducts and final object that
satisfies three technical properties: $1)$ a form of partial pairing
for compatible partial maps; $2)$ disjointness of coprojections; and
$3)$ joint monicity of partial projections. In effectuses,
predicates are total maps of the form $X\rightarrow Y+1$ which, as
usual, are equivalent to partial maps of the form $X\rightarrow Y$. In
particular, given an effectus $B$, the category $Par(B)$ of partial
maps over $B$ is enriched over the category of PCMs.

%

Interestingly, the notion of homomorphism for effect
algebras~\cite{Cho15,Effectus15,Jacobs18} is similar to our notions of
PCM morphisms. Indeed, the similarities between our
Definition~\ref{def:seprel} and Definition~\ref{def:pcmmorph} with
\cite[Definition 12]{Effectus15} are clear. 
One difference, however, is
that 
their definition only considers PCM morphisms with trivial \seprel
(what we call total morphisms), whereas our morphisms can have more
general \seprels, and are thus properly partial. 
%
%
The origins of PCM morphisms, as described in our paper, lie in
separation logic and we have explored their applications to
verification of concurrent programs. Nevertheless the close relation
to effectuses encourages us to explore future applications of our work
to recent extensions of separation and Hoare logic such as quantum
relational Hoare logic~\cite{Unruh+popl19,Unruh+lics19}, relational
proofs of quantum programs~\cite{BartheHYYZ+popl20}, and probabilistic
separation (and other program)
logics~\cite{BatzKKMN+popl19,BartheHL+popl20,SatoABGGH+popl19,TassarottiH+popl19}.

\section{Conclusion and Future Work}
\label{s:concl}
Morphisms are a standard notion in algebra and category theory, where
algebraic structures give rise to structure-preserving functions,
i.e. morphisms, between them. We adapt the notion of morphisms to the
structure of PCMs, thereby extending standard algebraic and
categorical approaches to concurrent separation logics.

The mathematics behind this adaptation gives rise to \seprels, which
delineate the domain where a function is structure preserving and thus
a morphism. We introduce invertibility as a property of morphisms and
\seprels that allows working with morphisms under abstraction. Our
exposition of PCMs and their morphisms is natural; we recover the
standard algebraic constructions (e.g. that of a sub-object, a
sub-PCM), show that the constructions preserve structure (e.g.,
composition of morphisms is a morphism, equalizer of morphisms is a
\seprel, etc.), and show that invertibility is preserved under
composition and products of morphisms.
Morphisms are
useful in specs to compute values out of the
state; structure preservation ensures that morphisms are well behaved
under ownership transfer.

In the future, we will build on the scaffolding provided by PCM
morphisms, along with resource morphisms and
simulations~\cite{NanevskiBDF+oopsla19}, to obtain an algebraic theory
of \emph{linearizable} resources. Such a formalism will unite logical,
categorical, and type-theoretic
foundations~\cite{harper:trinitarianism}, while supporting the
verification of a wide range of realistic concurrent programs.



\begin{acks}
We thank Gordon Stewart and Joe Tassarotti for their comments on
various drafts of the paper. We thank the anonymous reviewers from the
POPL'21 PC and AEC for their feedback. This research was partially
supported by the Spanish MICINN projects BOSCO (PGC2018-102210-B-I00)
and ProCode-UCM (PID2019-108528RB-C22),
the European Research Council
project Mathador (ERC2016-COG-724464) and the US National Science
Foundation (NSF). Any opinions, findings, and conclusions or
recommendations expressed in the material are those of the authors and
do not necessarily reflect the views of the funding agencies.
\end{acks}


\bibliography{bibmacros,references,proceedings}

\ifappendix
\clearpage
\appendix
\newpage
\section{Proof outline for $\progUnlock$}\label{s:apxunlock}

In this appendix, we present the proof outline for the following spec
and code for $\progUnlock$.
\[
\begin{array}[t]{rcl}
 \progUnlock & : & {\spec{\lambda s \ldot \fapp{\cs{\alpha}}{s} = \own \land \relAlpha \cs{s} \co{s}}}
                   {\spec{\lambda s \ldot \fapp{\cs{\alpha}}{s} = \nown \land \relAlpha \cs{s} \co{s}}}\, @\, {\rTicket}\\
             & = & \left\langle \incfetch(\var{dsp}) ; \ghostcode{
                        \unlockTr                       }\right\rangle
\end{array}
\]
Here $\alpha$ and $\rel{\alpha}$ are the morphism and its \seprel as
defined in Section~\ref{s:morphabst}. The predicates $\ord$ and
$\nogaps$ have been defined in Section~\ref{s:transitions}.

\newcounter{unlockcodea}
\def\lineno{\stepcounter{unlockcodea}\textsc{\theunlockcodea}}
\[\begin{array}[t]{r@{\quad}l}
\lineno. & \proofspec{
              \fapp{\cs{\alpha}}{s} = \own \wedge \cs{s} \rel{\alpha} \co{s}
           }\\
\lineno. & \proofspecL{\!\!\!\begin{array}[t]{l}
             \exists t\ldot \papp{(\sgS{s})}{t} = \currentT \wedge (\forall t'\ldot t' \in \dom{\sgS{s}}\ldot t' \neq t \Rightarrow \papp{(\sgS{s})}{t'} \neq \currentT) \wedge \hbox{}\\
             \qquad \papp{\ord}{\sgT{s}} \wedge \papp{\nogaps}{\sgT{s}} \wedge \cs{s} \rel{\alpha} \co{s} \}
             \end{array}
           }\\
\lineno. & \proofspecL{\!\!\!\begin{array}[t]{l}
             \papp{(\sgS{s})}{\psiT{s}} = \currentT \wedge (\forall t'\ldot t' \in \dom{\sgS{s}}\ldot t' \neq \psiT{s} \Rightarrow \papp{(\sgS{s})}{t'} \neq \currentT)  \wedge \hbox{}\\
             \qquad \papp{\ord}{\sgT{s}} \wedge \papp{\nogaps}{\sgT{s}} \wedge \cs{s} \rel{\alpha} \co{s} \}
             \end{array}
           }\\
\lineno. & \left\langle \incfetch(\var{dsp}) ; \ghostcode{
                        \unlockTr       
                }\right\rangle\\
\lineno. & \proofspec{
            (\forall t'\ldot t' \in \dom{\sgS{s}}\ldot \papp{(\sgS{s})}{t'} \neq \currentT)
             \wedge \cs{s} \rel{\alpha} \co{s}
           }\\
\lineno. & \proofspec{
            \fapp{\cs{\alpha}}{s} = \nown \wedge \cs{s} \rel{\alpha} \co{s}
           }\\
\end{array}\]

Line 1 is the precondition we want for $\progUnlock$. 
The first conjunct in line 2 derives by unfolding the definition of
$\alpha$. The second conjunct follows from
$\cs{s} \rel{\alpha} \co{s}$, which tells us that we can only have one
$\currentT$ ticket in $\sgS$ (and in $\sgO$ but that doesn't matter at
this point). The third and fourth conjunct just materialize because
they're part of the state space so we can always assume them. The
fifth conjunct propagates by stability of $\rel{\alpha}$.

Now, in line 3 we make the inference that $t$ must equal
$\psiT{s}$. This follows by the ordering property given by
$\papp{\ord}{\sgT{s}}$, by no-gap property given by
$\papp{\nogaps}{\sgT{s}}$, and by the uniqueness property for
$\currentT$ tickets given by $\cs{s} \rel{\alpha} \co{s}$. Indeed, if
$\papp{(\sgS{s})}{t} = \currentT$, then by the ordering property we
know that $t$ must be bigger than all $\servedT$ tickets. Moreover,
$t$ is smaller than all $\waitingT$ tickets and $t$ is the unique
$\currentT$ ticket. Thus, there are no other tickets between $t$ and
the largest $\servedT$ ticket. As there are no gaps in tickets, it must
be $t = \psiT{s}$.

In line 5 we get that $\sgS{s}$ has no $\currentT$ tickets anymore,
because $\unlockTr$ transition switched the ticket $t$ to $\servedT$,
and it was the only $\currentT$ ticket in $\sgS{s}$.  The
$\rel{\alpha}$ conjuncts propagate because it's preserved by both our
steps and other steps. The $\ord$ and $\nogaps$ conjuncts
also propagate, but we don't need them anymore, so we elide them.

Line 6 simply follows by definition of $\alpha$.

\section{\Seprels and stability under transitions}\label{s:stability}

In Section~\ref{s:example}, we carried out the proof outline by 
appealing to the stability of the invariant $\relAlpha {} {}$ under the 
transitions of the resource $\rTicket$.  Here, we explain the notion 
of stability under a resource in more detail and prove that it holds 
of $\relAlpha {} {}$. 

First, let us have a look at the interference aspect of our 
logic. Other threads can concurrently perform transitions on a resource 
as long as the preconditions of each transition are met.  For example, 
other threads working with the resource $\rTicket$ can acquire new 
tickets by taking transition $\taketxTr$. 
From the point of view of a thread, a transition $t$ taken by other 
thread is \emph{transposed} relatively to the same transition taken by 
the thread itself. We denote such transposed transition by 
$\tp{t}$. What transposition involves is simply swapping the 
\emph{self} and \emph{other} components of the transition states, that 
is: 
\[
\tp{t}\,s\,s' = t\,(\co{s},\cs{s})\,(\co{s'},\cs{s'}) 
\]
Intuitively, a transition can only change its \emph{self} state and 
must leave the \emph{other} state untouched. A transposed transition 
then changes the \emph{self} state of another thread, which we view as 
\emph{other} state. 

For example, recall the definition of the transition $\taketxTr$ in 
Section~\ref{s:example}. 
\begin{align*}
\taketxTr~s~s' & \eqdef   
  \fapp{\sigmaS}{s'} = \{\fresh{\fapp{\sigmaT}{s}}\} \cupdot  \fapp{\sigmaS}{s}
\end{align*}
\noindent The transposition $\taketxTr$ is then: 
\begin{align*}
\tp{\taketxTr}~s~s' & \equiv 
  \fapp{\sigmaO}{s'} = \{\fresh{\fapp{\sigmaT}{s}}\} \cupdot  \fapp{\sigmaO}{s}
\end{align*}
\noindent We can see that the transposed version says that another 
thread acquired a fresh ticket over all tickets both in its 
\emph{self} an \emph{other} state and added it to its \emph{self}
state, which we see as a part of our \emph{other} state.

\subsection{Stability of state space}

In Section~\ref{s:example}, Figure~\ref{fig:ststicket}, we defined the state
space $\STicket$ by imposing two properties, $\ord$ and $\nogaps$, on cartesian
product $(\mapTctl) \times (\mapTctl)$. This leaves us with an obligation to
prove that the transitions of $\rTicket$ preserve 
these properties, or else they won't be transitions of $\rTicket$.


Firstly,
notice also that $\ord$ is preserved by all the transitions of 
$\rTicket$. 

\begin{proposition}
Let $s, s' \in \STicket$ and $t$ a transition such that $\fapp{\ord}{s}$.
If $\fapp{\fapp{t}{s}}{s'}$ then $\fapp{\ord}{s'}$.
\end{proposition}
\begin{proof}
By case analysis on transition $t$.
\begin{itemize}
\item $\taketxTr$ introduces a fresh ticket as a $\waitingT$
  ticket. Thus, it preserves the ordering imposed by 
  $\ord$. 

\item $\lockTr$ starts from the precondition that the displayed ticket 
  is $\waitingT$. By the definition of $\psiT$, the displayed ticket 
  is one larger than the highest $\servedT$ ticket. Thus, by $\ord$, 
  there isn't a $\currentT$ ticket in the system, and the displayed 
  ticket is the smallest $\waitingT$. Therefore, the transition's 
  switching this ticket to $\currentT$ preserves $\ord$. 

\item $\unlockTr$ starts from the precondition that the displayed 
  ticket is $\currentT$. By $\ord$, this ticket is larger than all 
  $\servedT$ tickets, and smaller than all $\waitingT$
  tickets. Therefore, the transition's switching this ticket to 
  $\servedT$ preserves $\ord$. 
\end{itemize}

\end{proof}


Secondly, $\nogaps$ is also preserved by the transitions of \rTicket:

\begin{proposition}
Let $s, s' \in \STicket$ and $t$ a transition such that $\fapp{\nogaps}{s}$.
If $\fapp{\fapp{t}{s}}{s'}$ then $\fapp{\nogaps}{s'}$.
\end{proposition}
\begin{proof}
It's easy to see that $\nogaps$ is also preserved by the 
transitions, because fresh tickets are drawn in order, and are never 
discarded. 
\end{proof}


Notice that we didn't need to use $\nogaps$ to argue the 
preservation of $\ord$. This shows that $\nogaps$ isn't 
really necessary at the stage when we proof intermediate types
in Section~\ref{s:transitions}.
It is required only later in Section~\ref{s:morphabst} in the proof 
outline for the abstract version.

\subsection{Stability of proof invariants}

To be able to work with any property in a proof outline, the property 
must be preserved under interference of other threads, i.e., under 
sequences of transposed transitions of a resource. We call such 
properties \emph{stable}. The \seprel $\relAlpha {} {}$ is stable, which is 
why we could use it in the proof outline in Section~\ref{s:example}. 

\begin{proposition}
\label{prop:invostable}
Let $s$ be a state and $s'$ be s state after performing a sequence of 
transition or transposed transitions of \rTicket. 
If 
$\relAlpha \cs{s} \co{s}$ then
$\relAlpha \cs{s'} \co{s'}$. 
\end{proposition}
\begin{proof}
By case analysis of the transition of $\rTicket$. The transition 
\taketxTr and \unlockTr decreases the number of tickets labeled $\currentT$
and thus do not invalidate $\relAlpha {} {}$.

To see that $\relAlpha {} {}$ is stable under 
\lockTr and $\tp{\lockTr}$ as well, 
assume that $\relAlpha {} {}$ holds in $s$. 
From the definition of $\STicket$ it follows that
$\fapp{\ord}{\ct{s}}$ and we obtain $\numCurrent{\ct{s}} = 0$. Hence 
$\numCurrent{\ct{s'}} = 1$ and
$\relAlpha \cs{s'} \co{s'}$. 
\end{proof}

Finally, in the proof outlines for concrete specs, we crucially used the fact
that others cannot proceed with unlocking past a ticket that we hold and that
is labeled $\waitingT$, which is also a stability property:

\begin{proposition}
	Let $t$ be a ticket and $s, s'$ states such that
	$\papp{(\sgS{s})}{t} = \waitingT \land \psiT{s} \leq t$.
	If $s$ steps to $s'$ by a sequence of transposed transitions then 
	$\papp{(\sgS{s'})}{t} = \waitingT \land \psiT{s'} \leq t$.
\end{proposition}
\begin{proof}
  Notice that other 
  threads can increment the display, but not beyond $t$, since only 
  the owner of the ticket $t$ can increment the display to $t+1$. 
\end{proof}

\fi
\end{document}